\documentclass[a4paper,USenglish]{article}

\usepackage[utf8]{inputenc}
\usepackage[USenglish]{babel}
\usepackage{booktabs}
\usepackage{paralist}
\usepackage{microtype}
\usepackage{dsfont}
\usepackage{mathtools}
\usepackage{tabularx}
\usepackage{a4wide}
\usepackage{nicefrac}

\usepackage{tikz}
\usetikzlibrary{decorations.pathmorphing}
\usetikzlibrary{trees}
\usetikzlibrary{calc,backgrounds}
\usetikzlibrary{arrows.meta}
\tikzset{
    none/.style={color=white},
    smallvertex/.style={circle,draw,fill=black, inner sep= .75pt},
    stopvertex/.style={draw,circle,thick,dotted,preaction={draw,solid,orange,-,line width=0.80pt,double=orange,double distance=2\pgflinewidth,},inner sep=3.7pt},
    endvertex/.style={draw,circle,thick,dashed,preaction={draw,solid,blue,-,line width=0.80pt,double=blue,double distance=2\pgflinewidth,},inner sep=3.7pt},
    vertex/.style={circle,draw,fill=black, inner sep= 1.5pt},
    hollow/.style={circle,draw,inner sep=1pt},
    node/.style={color=black,circle,draw,dotted,very thick, inner sep= 1.5pt, },
    arc/.style={->,> = latex', thick},
    dashedarc/.style={->,> = latex', dashed, line width=.75pt},
    dashdotarc/.style={->,> = latex', dash dot, line width=.75pt},
    bluearc/.style={preaction={draw,blue,-,double=blue,double distance=.75pt}},
    orangearc/.style={preaction={draw,orange,-,double=orange,double distance=.75pt}},
    edge/.style={-,line width=.75pt},
}

\usepackage[vlined,linesnumbered,ruled]{algorithm2e}
\SetEndCharOfAlgoLine{}
\SetCommentSty{textrm}
\SetKwProg{Fn}{Function}{}{}

\usepackage{amsthm}
\theoremstyle{plain}
\newtheorem{theorem}{Theorem}
\newtheorem{lemma}{Lemma}

\newtheorem{observation}{Observation}
\newtheorem{proposition}{Proposition}
\newtheorem{rrule}{Reduction Rule}[section]
\theoremstyle{definition}
\newtheorem{definition}{Definition}

\theoremstyle{remark}

\usepackage[square,sectionbib,numbers]{natbib}
\usepackage{xcolor}

\usepackage[pagebackref,pdfdisplaydoctitle,menucolor=orange!40!black,filecolor=magenta!40!black,urlcolor=blue!40!black,linkcolor=red!40!black,citecolor=green!40!black,colorlinks]{hyperref}

\usepackage[nameinlink,sort&compress]{cleveref}
\crefname{rrule}{Reduction Rule}{Reduction Rules}
\crefname{construction}{Construction}{Constructions}
\crefname{claim}{Claim}{Claims}
\crefname{paragraph}{Paragraph}{Paragraphs}
\crefname{observation}{Observation}{Observations}
\crefname{theorem}{Theorem}{Theorems}
\crefname{lemma}{Lemma}{Lemmata}
\crefname{proposition}{Proposition}{Propositions}
\crefname{corollary}{Corollary}{Corollaries}
\crefname{remark}{Remark}{Remarks}
\crefname{section}{Section}{sections}
\crefname{chapter}{Chapter}{Chapters}
\crefname{figure}{Figure}{Figures}
\crefname{table}{Table}{Tables}
\crefname{definition}{Definition}{Definitions}
\crefname{algorithm}{Algorithm}{Algorithms}
\crefname{equation}{equation}{equations}
\crefname{ineq}{inequality}{inequalities}
\crefname{appendix}{Appendix}{Appendices}

\newcommand{\taskdef}[3]{
	\begin{center}
		\begin{minipage}{0.95\textwidth}
			\textsc{#1}\\[2mm]
			\setlength{\tabcolsep}{3pt}
			\begin{tabularx}{\textwidth}{@{}lX@{}}
				\normalsize \textbf{Input:} 	& \normalsize #2 \\
				\normalsize \textbf{Task:} 	& \normalsize #3
			\end{tabularx}
		\end{minipage}
	\end{center}
}

\newcommand{\bc}{betweenness centrality}
\newcommand{\bctask}{\textsc{Betweenness Centrality}}
\newcommand{\maxpath}{maximal induced path}
\newcommand{\maxpaths}{maximal induced paths}

\newcommand{\N}{\mathds{N}}
\DeclareMathOperator{\Left}{left}
\DeclareMathOperator{\Mid}{mid}
\DeclareMathOperator{\Right}{right}

\newcommand{\Vone}{\ensuremath{V^{=1}}}
\newcommand{\Vtwo}{\ensuremath{V^{=2}}}
\newcommand{\Vthree}{\ensuremath{V^{\ge3}}}
\newcommand{\centrality}{\ensuremath{C_B}}
\newcommand{\BC}{\ensuremath{BC}}
\newcommand{\Wleft}{\ensuremath{W^{\Left}}}
\newcommand{\Wright}{\ensuremath{W^{\Right}}}
\newcommand{\Wtable}{\ensuremath{W}}
\DeclareMathOperator{\Pend}{Pen}
\DeclareMathOperator{\Inc}{Inc}
\DeclareMathOperator{\mmod}{mod}

\newcommand{\xleft}{\ensuremath{x^{\Left}}}
\newcommand{\xmid}{\ensuremath{x^{\Mid}}}
\newcommand{\xright}{\ensuremath{x^{\Right}}}
\newcommand{\Xleft}{\ensuremath{X^{\Left}}}

\newcommand{\Xright}{\ensuremath{X^{\Right}}}

\newcommand{\Pmax}{\ensuremath{P^{\max}}}
\newcommand{\AllPmax}{\ensuremath{\mathcal{P}^{\max}}}

\newcommand{\qq}{\ensuremath{q}}
\newcommand{\qr}{\ensuremath{r}}

\title{\LARGE \bf An Adaptive Version of Brandes' Algorithm\\ for Betweenness Centrality\footnote{Work partially supported by DFG Project FPTinP, NI 369/16.} \footnote{An extended abstract of this work appears in the proceedings of the 29th International Symposium on Algorithms and Computation (ISAAC~'18).}}

\newcommand{\tuaddress}{Algorithmics and Computational Complexity, Faculty IV, TU Berlin, Germany}
\author{Matthias~Bentert \and Alexander~Dittmann \and Leon~Kellerhals \and André~Nichterlein \and Rolf~Niedermeier}
\date{\tuaddress\\
\texttt{\small \{matthias.bentert,leon.kellerhals,andre.nichterlein,rolf.niedermeier\}@tu-berlin.de alexander.dittmann@campus.tu-berlin.de}}

\begin{document}

\maketitle

\begin{abstract}
Betweenness centrality---measuring how many shortest paths pass through a vertex---is one of the most important network analysis concepts for assessing the relative importance of a vertex.
The well-known algorithm of Brandes [J. Math. Sociol.~'01] computes, on an $n$-vertex and $m$-edge graph, the betweenness centrality of all vertices in $O(nm)$ worst-case time. 
In later work, significant empirical speedups were achieved by preprocessing degree-one vertices and by graph partitioning based on cut vertices.
We contribute an algorithmic treatment of degree-two vertices, which turns out to be much richer in mathematical structure than the case of degree-one vertices.
Based on these three algorithmic ingredients, we provide a strengthened worst-case running time analysis for betweenness centrality algorithms.
More specifically, we prove an adaptive running time bound~$O(kn)$, where~$k < m$ is the size of a minimum feedback edge set of the input graph.
\end{abstract}

%
%
\section{Introduction}
One of the most important building blocks in network analysis is to determine a vertex's relative importance in the network.
A key concept herein is \emph{\bc} as introduced in~1977 by \citet{Fre77}; it measures centrality based on shortest paths. Intuitively, for each vertex, \bc{} counts the (relative) number of shortest paths that pass through the vertex. 
A straightforward algorithm for computing the betweenness centrality on undirected (unweighted) $n$-vertex graphs runs in $O(n^3)$~time, and improving this to~$O(n^{3-\varepsilon})$ time for any~$\varepsilon > 0$ would break the so-called APSP-conjecture~\cite{AGW15}.
In~2001, Brandes~\cite{Bra01} presented the to date theoretically fastest algorithm, improving the running time to $O(nm)$ for graphs with $m$~edges. 
As many real-world networks are sparse, this is a far-reaching improvement, having a huge impact also in practice.
We remark that Newman \cite{Newman2004,New10} presented a high-level description of an algorithm computing a variant of \bc{} which also runs in~$O(nm)$ time.

Since betweenness centrality is a measure of outstanding importance in network science, it finds numerous applications in diverse areas,
e.g. in 
social network analysis~\cite{New10,WF94} or neuroscience~\cite{KEB15,Med17}.  
Provably speeding up betweenness centrality computations is the ultimate goal of our research.
To this end, we extend previous work and provide a rigorous mathematical analysis that yields a new (parameterized) running time upper bound of the corresponding algorithm.

Our work is in line with numerous research efforts concerning the development of 
algorithms for computing betweenness centrality, including approximation algorithms~\cite{BKMM07,GSS08,RK16}, parallel and distributed algorithms~\cite{TTS09,WT13}, streaming and incremental algorithms~\cite{GMB12,NPR14}, and exact~\cite{EIBT15} and fixed-parameter algorithms~\cite{CDP19}.
Formally, we study the following problem:
\taskdef{Betweenness Centrality}
{An undirected graph~$G$.}
{Compute the \emph{betweenness centrality}~$C_B(v) :=
\sum_{s, t \in V(G)} \nicefrac{\sigma_{st}(v)}{\sigma_{st}}$ for each vertex~$v \in V(G)$.}
Herein, $\sigma_{st}$ is the number of shortest paths in~$G$ from vertex~$s$ to vertex~$t$, and $\sigma_{st}(v)$ is the number of shortest paths from~$s$ to~$t$ that additionally pass through~$v$.%
\footnote{To simplify matters, we set~$\sigma_{st}(v) = 0$ if~$v=s$ or~$v=t$. This is equivalent to the definition used by \citet{Bra01} but differs from the definition used by \citet{Newman2004}, where~$\sigma_{s t}(s) = 1$.}

Extending previous, more empirically oriented work of \citet{BGPL12,PEZDB14}, and \citet{SKSC17} (see \cref{sec:overview} for a description of their approaches), our main result is an algorithm for \textsc{Betweenness Centrality} that runs in~$O(kn)$ time, where $k$ denotes the feedback edge number of the input graph~$G$.
The linear-time computable \emph{feedback edge number} of~$G$ is the minimum number of edges one needs to delete from~$G$ in order to make it a forest.\footnote{Notably, \textsc{Betweenness Centrality} computations have also been studied when the input graph is a tree~\cite{WT13}, hinting at the practical relevance of this special case.} 
Clearly,~$k=0$ holds on trees, and~$k \le m$ holds in general.
Thus our algorithm is \emph{adaptive}, i.e., it interpolates between linear time for constant~$k$ and the running time of the best unparameterized algorithm.\footnote{We mention in passing that in recent work~\cite{MNN17} we employed the same parameter ``feedback edge number'' in terms of theoretically analyzing known data reduction rules for computing maximum-cardinality matchings. Recent empirical work with this algorithm demonstrated  significant accelerations of the state-of-the-art matching algorithm~\cite{KNNZ18}.
}
Obviously, one can compute~$k$ in linear time, using depth-first search; however~$k\approx m-n$, so we provide no asymptotic improvement over Brandes' algorithm for most graphs. 
When the input graph is very tree-like ($m=n+o(n)$), however, our new algorithm theoretically improves on Brandes' algorithm. 
Real-world networks showing the relation between PhD candidates and their supervisors~\cite{DMB11,Joh84} or the ownership relation between companies~\cite{NLGC02} typically have a feedback edge number that is smaller than the number of vertices or edges~\cite{NNUW13} by orders of magnitude.\footnote{The networks are available in the Pajek Dataset of Vladimir Batagelj and Andrej Mrvar (2006) (\url{http://vlado.fmf.uni-lj.si/pub/networks/data/}).}
Moreover, \mbox{\citet{BGPL12}}, building on Brandes' algorithm and basically shrinking the input graph by deleting degree-one vertices in a preprocessing step, report on significant speedups in comparison with Brandes' basic algorithm in empirical tests with real-world social networks. 
For roughly half of their networks, $m-n$ is smaller than~$n$ by at least one order of magnitude.

Our algorithmic contribution is to complement the works of \citet{BGPL12}, \citet{PEZDB14}, and \citet{SKSC17} by, roughly speaking, additionally dealing with degree-two vertices. 
These vertices are much harder to cope with and to analyze since, other than degree-one vertices, they may lie on shortest paths between two vertices.
From a practical point of view, one may expect a significant speedup if one can take care of degree-two vertices more quickly.
This is due to the nature of many real-world social networks having a power-law degree distribution \cite{barabasi99}; thus a large fraction of the vertices are of degree one or two.
Recently, \mbox{\citet{VBC18}} used a heuristic approach to process degree-two vertices for improving the performance of their \textsc{Betweenness Centrality} algorithms on several real-world networks.

Our work is purely theoretical in spirit.
Our most profound contribution is to analyze the worst-case running time of the proposed betweenness centrality algorithm based on degree-one-vertex processing~\cite{BGPL12}, usage of cut vertices~\cite{PEZDB14,SKSC17}, and our degree-two-vertex processing.
To the best of our knowledge, this provides the first proven worst-case ``improvement'' over Brandes' upper bound in a relevant special case.  

\paragraph{Notation.}
We use mostly standard graph notation.
Given a graph~$G$,~$V(G)$ and~$E(G)$ denote the vertex respectively edge set of~$G$ with~$n=|V(G)|$ and~$m = |E(G)|$.
We denote the vertices of degree one, two, and at least three by~$\Vone(G)$, $\Vtwo(G)$, and~$\Vthree(G)$, respectively.
A \emph{cut vertex} or \emph{articulation vertex} is a vertex whose removal disconnects the graph.
A connected component of a graph is \emph{biconnected} if it does not contain any cut vertices, and hence, no vertices of degree one.
A path~$P = v_0 \dots v_\qq$ is a graph with~$V(P) = \{v_0, \ldots, v_\qq\}$ and~$E(P) = \{\{v_i,v_{i+1}\} \mid 0 \le i < \qq \}$.
The length of the path~$P$ is $|E(P)|$.
Adding the edge~$\{v_\qq,v_0\}$ to~$P$ gives a cycle~$C = v_0 \dots v_\qq v_0$.
The distance~$d_G(s,t)$ between vertices~$s,t \in V(G)$ is the length of the shortest path between~$s$ and~$t$ in~$G$.
The number of shortest $s$-$t$--paths is denoted by~$\sigma_{st}$.
The number of shortest $s$-$t$--paths containing some vertex~$v$ is denoted by~$\sigma_{st}(v)$.
We set~$\sigma_{st}(v) = 0$ if~$s=v$ or~$t=v$ (or both).

We set $[j,k] := \{j, j+1,\ldots, k\}$ and denote for a set~$X$ by~$\binom{X}{i}$ the size-$i$ subsets of~$X$.

\paragraph{Paper outline.}
The presentation of our algorithm is split into two parts:
In \Cref{sec:overview} we present the strategy of our algorithm.
\Cref{sec:maxpaths} deals with the main technical challenge of our algorithm, namely how to deal with consecutive degree-two vertices.
Some proofs in the latter part are deferred to the appendix.

Finally, we conclude in \cref{sec:conclusion}.

\section{Algorithm Overview}
\label{sec:overview}
In this section, we review our algorithmic strategy to compute the betweenness centrality of
each vertex.
Before doing so, since we build on the works of \mbox{\citet{Bra01,BGPL12}}, \citet{PEZDB14}, and \citet{SKSC17}, we first give the high-level ideas behind their algorithmic approaches. 
Then, we describe the ideas behind our extension.
We assume throughout our paper that the input graph is connected. 
Otherwise, we can process the connected components one after another. 

\paragraph{Existing algorithmic approaches.} 
\citet{Bra01} developed an $O(nm)$-time algorithm which essentially runs modified breadth-first searches (BFS) from each vertex of the graph. 
In each of these modified BFS starting in a vertex~$s$, Brandes' algorithm computes the ``effect'' that~$s$ has on the betweenness centrality values of all other vertices. 
More formally, the modified BFS starting at vertex~$s$ computes for every~$v \in V(G)$ the value
\[\sum_{t \in V(G)} \frac{\sigma_{st}(v)}{\sigma_{st}}.\]

Reducing the number of performed modified BFS in Brandes' algorithm is one way to speed up Brandes' algorithm.
To this end, a popular approach is to remove in a preprocessing step all degree-one vertices from the graph~\cite{BGPL12,PEZDB14,SKSC17}.
By repeatedly removing degree-one vertices, whole ``pending trees'' (subgraphs that are trees and are connected to the rest of the graph by a single edge) can be deleted.
Considering a degree-one vertex~$v$, observe that in \emph{each} shortest path~$P$ starting at~$v$, the second vertex in~$P$ is the single neighbor~$u$ of~$v$.
Hence, after deleting~$v$, one needs to store the information that~$u$ had a degree-one neighbor.
To this end, one uses for each vertex~$w$ a counter called~$\Pend[w]$ (for \emph{pending}) that stores the number of vertices in the subtree pending on~$w$ that were deleted before.
In contrast to e.\,g.\ \citet{BGPL12}, we initialize for each vertex~$w \in V$ the value~$\Pend[w]$ with one instead of zero (so we count~$w$ as well).
This simplifies most of our formulas.
See \Cref{fig:pen} (Parts (1.)~to (3.))~for an example of the~$\Pend[\cdot]$-values of the vertices at different points in time.
\begin{figure}
	\centering
	\begin{tikzpicture}[scale=1.25]
		\node at (1.5,1.55) {(1.)};
		\draw[rounded corners,draw=black!50,thick] (-0.35,1.35) rectangle (3.35,-0.35);
	
		\node[circle, draw,inner sep=2pt] at(0,1) (a) {1};
		\node[circle, draw,inner sep=2pt] at(0,0) (b) {1};

		\node[circle, draw,inner sep=2pt] at(1,1) (c) {1} edge (a);
		\node[circle, draw,inner sep=2pt] at(1,0) (d) {1} edge (b) edge (a);

		\node[circle, draw,inner sep=2pt] at(2,1) (e) {1} edge (d);
		\node[circle, draw,inner sep=2pt] at(2,0) (f) {1} edge (d) edge (e);

		\node[circle, draw,inner sep=2pt] at(3,1) (g) {1} edge (f);
		\node[circle, draw,inner sep=2pt] at(3,0) (h) {1} edge (f) edge (g);

 		\draw [->,decorate,decoration=snake,thick,draw=black!60] (3.5,0.5) -- (4.5,0.5);
 		\draw [->,decorate,decoration=snake,thick,draw=black!60] (4.5,0.00) -- (3.5,-1.00);
 		\draw [->,decorate,decoration=snake,thick,draw=black!60] (3.5,-1.5) -- (4.5,-1.5);
		
		\begin{scope}[xshift = 5cm]
		\draw[rounded corners,draw=black!50,thick] (-0.35,1.35) rectangle (3.35,-0.35);
			\node at (1.5,1.55) {(2.)};
			\node[circle, draw,inner sep=2pt] at(0,1) (a) {1};
			
			\node[circle, draw,inner sep=2pt] at(1,1) (c) {1} edge (a);
			\node[circle, draw,inner sep=2pt] at(1,0) (d) {2} edge (a);

			\node[circle, draw,inner sep=2pt] at(2,1) (e) {1} edge (d);
			\node[circle, draw,inner sep=2pt] at(2,0) (f) {1} edge (d) edge (e);

			\node[circle, draw,inner sep=2pt] at(3,1) (g) {1} edge (f);
			\node[circle, draw,inner sep=2pt] at(3,0) (h) {1} edge (f) edge (g);
		\end{scope}

		\begin{scope}[yshift = -2.2cm]
		\draw[rounded corners,draw=black!50,thick] (-0.35,1.35) rectangle (3.35,-0.35);
			\node at (1.5,1.55) {(3.)};
			\node[circle, draw,inner sep=2pt] at(1,0) (d) {4};

			\node[circle, draw,inner sep=2pt] at(2,1) (e) {1} edge (d);
			\node[circle, draw,inner sep=2pt] at(2,0) (f) {1} edge (d) edge (e);

			\node[circle, draw,inner sep=2pt] at(3,1) (g) {1} edge (f);
			\node[circle, draw,inner sep=2pt] at(3,0) (h) {1} edge (f) edge (g);
		\end{scope}

		\begin{scope}[xshift = 5cm, yshift = -2.2cm]
		\draw[rounded corners,draw=black!50,thick] (-0.35,1.35) rectangle (3.35,-0.35);
			\node at (1.5,1.55) {(4.)};
			\node[circle, draw,inner sep=2pt] at(0,0) (d) {4};

			\node[circle, draw,inner sep=2pt] at(1,1) (e) {1} edge (d);
			\node[circle, draw,inner sep=2pt] at(1,0) (f) {3} edge (d) edge (e);

			\node[circle, draw,inner sep=2pt] at(2,0) (ff) {6};

			\node[circle, draw,inner sep=2pt] at(3,1) (g) {1} edge (ff);
			\node[circle, draw,inner sep=2pt] at(3,0) (h) {1} edge (ff) edge (g);
		\end{scope}
	\end{tikzpicture}
	\caption{An initial graph where the $\Pend[\cdot]$-value of each vertex is~$1$ (top left) and the same graph after deleting one (top right) or both (bottom left) pending trees using Reduction Rule~\ref{rr:deg-one-vertices}. The labels are the respective~$\Pend[\cdot]$-values. Subfigure (4.)~shows the graph of (3.)~after applying Lemma~\ref{lem:split} to the only cut vertex of the graph.}
	\label{fig:pen}
\end{figure}
We obtain the following (weighted) problem variant. 
\taskdef{Weighted Betweenness Centrality}
{An undirected graph~$G$ and vertex weights~$\Pend\colon V(G) \rightarrow \N$.}
{Compute for each vertex~$v \in V(G)$ the weighted betweenness centrality 
\begin{align}
	C_B(v) := \sum_{s,t \in V(G)} \gamma(s, t, v), \label{eq:betw-cent-with-pend}
\end{align}
where~$\gamma(s, t, v) \coloneqq \Pend[s] \cdot \Pend[t] \cdot \sigma_{st}(v)/\sigma_{st}$.
}

The effect of a degree-one vertex to the betweenness centrality value of its neighbor is captured in the next data reduction rule. 
\begin{rrule}[\cite{BGPL12,PEZDB14,SKSC17}] 
	\label[rrule]{rr:deg-one-vertices}
	Let~$G$ be a graph, let $s \in V(G)$ be a degree-one vertex, and let~$v \in V(G)$ be the neighbor of~$s$. 
	Then increase~$\Pend[v]$ by~$\Pend[s]$, increase the betweenness centrality of~$v$ by~$\Pend[s]  \cdot \sum_{t\in V(G) \setminus \{s,v\}} \Pend[t]$, and remove~$s$ from the graph.
\end{rrule}
By \cref{rr:deg-one-vertices} the influence of a degree-one vertex to the betweenness centrality of its neighbor can be computed in constant time.
This is since
\[\sum_{t \in V(G) \setminus \{s, v\}} \Pend[t] = \Big( \sum_{t \in V(G)} \Pend[t] \Big) - \Pend[s] - \Pend[v],\]
and~$\sum_{t \in V(G)}\Pend[t]$ can be precomputed in linear time.

A second approach to speed up Brandes' algorithm is to split the input graph~$G$ into smaller connected components and process them separately~\cite{PEZDB14,SKSC17}. 
This approach is a generalization of the ideas behind removing degree-one vertices and works with cut vertices.
The basic observation for this approach is as follows:
Consider a cut vertex~$v$ such that removing~$v$ breaks the graph into two connected components~$C_1$ and~$C_2$ (the idea generalizes to more components).
Obviously, every shortest path~$P$ in~$G$ that starts in~$C_1$ and ends in~$C_2$ has to pass through~$v$.
For the betweenness centrality values of the vertices inside~$C_1$ (inside~$C_2$) it is not important where exactly~$P$ ends (starts).
Hence, for computing the betweenness centrality values of the vertices in~$C_1$, it is sufficient to know which vertices in~$C_1$ are adjacent to~$v$ and how many vertices are contained in~$C_2$.
Thus, in a preprocessing step one can just add to~$C_1$ the cut vertex~$v$ with~$\Pend[v]$ being increased by the sum of~$\Pend[\cdot]$-values of the vertices in~$C_2$ (see \cref{fig:pen} (bottom)).
Formally, this is done as follows.

\begin{lemma}[\cite{PEZDB14,SKSC17}]
\label[lemma]{lem:split}
Let~$G$ be a connected graph, let~$v$ be a cut vertex such that removing~$v$ yields~$\ell \ge 2$ connected components~$C_1, \dots, C_\ell$, and let~$\xi \coloneqq \Pend[v]$.
Then remove~$v$, add a vertex~$v_i$ to every component~$C_i$, make it adjacent to all vertices in the respective component that were adjacent to~$v$, and set
\[ \Pend[v_i] = \xi + \sum_{j \in [1,\ell]\setminus \{i\}} \sum_{w \in V(C_j) \setminus \{v_j\}} \Pend[w]. \]
For a vertex~$v$ in component~$C_i$ denote by~$\centrality^{C_i}(v)$ the \bc{} of~$v$ within the component~$C_i$.
Computing the \bc{} of each connected component independently, increasing the \bc{} of~$v$ by
\[ \sum_{i=1}^{\ell} \big( C_B^{C_i}(v_i) + (\Pend[v_i] - \xi) \cdot \sum_{s \in V(C_i) \setminus \{v_i\}} \Pend[s] \big), \]
and ignoring all new vertices~$v_i$ is the same as computing the \bc{} in~$G$, that is,
\[
	C_B^G(u) = \begin{cases}
		C_B^{C_i}(u), & \hspace{-1.25cm}\text{if } u \in V(C_i) \setminus \{v_i\}; \\
		\sum_{i = 1}^{\ell} \big ( C_B^{C_i}(v_i) + (\Pend[v_i] - \xi) \cdot \sum_{s\in V(C_i) \setminus\{v_i\}} \Pend[s] \big ), & \hspace{0.40cm}\text{if } u = v.
	\end{cases}
\]
\end{lemma}

Applying the above procedure as preprocessing on all cut vertices and degree-one vertices leaves us with biconnected components that we can each solve independently.
We first look at the special case that the connected component (obtained from \cref{lem:split}) is a cycle, then deal with the general, more involved case.

\subsection{Dealing with Cycles}
We now show how to solve \textsc{Weighted Betweenness Centrality} on cycles with a linear-time dynamic programming algorithm.
Note that the vertices in the cycle can have different \bc{} values as they can have different~$\Pend[\cdot]$-values.

\begin{proposition}
	\label[proposition]{prop:pendingcycles}
	Let~$C = x_0 \dots x_\qq x_0$ be a cycle.
	Then, one can compute the weighted \bc{} of the vertices in~$C$ in~$O(\qq)$ time and space.
\end{proposition}
\begin{proof}
	We first introduce some notation needed for the proof, then we will show how to
	compute~$\BC[v]$ for~$v \in V(C)$ efficiently. Finally, we prove the running time.

	By~$[x_i, x_j]$, $0 \le i, j \le \qq$ we denote the set of vertices~$\{x_i, x_{i+1 \mmod (\qq+1)},\allowbreak x_{i+2 \mmod (\qq+1)}, \dots, x_j\}$.
	For a \maxpath{}~$\Pmax = x_0 \ldots x_\qq$ we define
	\begin{align*}
		\Wleft[x_i] &:= \sum_{k=0}^{i} \Pend[x_i],\; \text{and} \\
		W[x_i, x_j] &:= \begin{cases}
			\Pend[x_i], & \text{if } i=j;\\
			\Wleft[x_j] - \Wleft[x_i] + \Pend[x_i], & \text{if } i < j;\\
			\Wleft[x_\qq] - \Wleft[x_i] + \Wleft[x_j] + \Pend[x_i], & \text{if } i > j.\\
		\end{cases}
	\end{align*}
	The value~$W[x_i, x_j]$ is the sum of all~$\Pend$-values from~$x_i$
	to~$x_j$, clockwise.
	Further, we denote by~$\varphi(i) = (\frac{\qq+1}{2} + i) \mmod (\qq+1)$ the index that is
	``opposite'' to~$i$ on the cycle.
	Note that if~$\varphi(i) \in \mathds{N}$, then~$x_{\varphi(i)}$ is the unique vertex
	in~$C$ to which there are two shortest paths from~$x_i$, one
	visiting~$x_{i+1 \mmod (\qq+1)}$ and one visiting~$x_{i-1 \mmod (\qq+1)}$.
	Otherwise, if~$\varphi(i) \not \in \mathds{N}$, then there is only one shortest path
	from~$x_k$ to any~$t \in V(C)$.
	For the sake of convenience in the next parts of the proof,
	if~$\varphi(i) \not \in \mathds{N}$, we say that~$\Pend[x_{\varphi(x_i)}] = 0$.
	Further, by~${\varphi^{\Left}(i) = \lceil\varphi(i)\rceil - 1 \mmod (\qq+1)}$ we denote the
	index of the vertex to the left of index~$\varphi(i)$
	and by~${\varphi^{\Right}(i) = \lfloor\varphi(i)\rfloor + 1 \mmod (\qq+1)}$ the index of the
	vertex to the right of index~$\varphi(i)$.

	We now describe the algorithm.
	For every vertex~$x_k$,~$0 \le k \le \qq$, we need to compute
	\[
		\BC[x_k] \coloneqq \sum_{s, t \in V(C)} \gamma(s, t, x_k)
		= \sum_{i=0}^\qq \sum_{t \in V(C)} \gamma(x_i, t, x_k).
	\]
	We determine these values with a dynamic program.
	We first compute~$\BC[x_0]$ and use it as the base case.
	Afterwards we show how to compute~$\BC[x_{k+1}]$ for~$0 \le k < \qq$ given the value of~$\BC[x_k]$.

	Towards computing~$\BC[x_0]$, observe that~$\gamma(x_i, t, x_0) = 0$
	if~$x_i = x_0$ or~$t = x_0$.
	Also, for every shortest path starting in~$x_{\varphi(0)}$ and ending
	in some~$x_j$,~$1 \le j \le \qq$, it holds
	that~$d_C(x_{\varphi(0)}, x_j) < d_C(x_{\varphi(0)}, x_0)$.
	Thus there is no shortest path starting in~$x_{\varphi(0)}$ that visits~$x_0$.
	Hence, we do not need to consider the cases~$i = 0$ or~$i = \varphi(0)$ and we have
	\begin{align*}
		\BC[x_0]
		&= \sum_{\substack{i=0\\0 \ne i\ne \varphi(0)}}^{\qq} \sum_{t \in V(C)} \gamma(x_i, t, x_0)\\
		&= \sum_{i=1}^{\varphi^{\Left}(0)} \sum_{t \in V(C)} \gamma(x_i, t, x_0) + \sum_{i = \varphi^{\Right}(0)}^{\qq} \sum_{t \in V(C)} \gamma(x_i, t, x_0)\\
		&= \sum_{i=1}^{\varphi^{\Left}(0)} \sum_{t \in V(C)} \Pend[x_i] \cdot \Pend[t] \cdot\frac{\sigma_{x_i t}(x_0)}{\sigma_{x_i t}}\\
		&\phantom{=}~+ \sum_{i = \varphi^{\Right}(0)}^{\qq} \sum_{t \in V(C)} \Pend[x_i] \cdot \Pend[t] \cdot \frac{\sigma_{x_i t}(x_0)}{\sigma_{x_i t}}
	\end{align*}
	By definition of~$\varphi(i)$ we have
	that~$d_C(x_i, x_{\varphi^{\Left}(i)}) =d_C(x_i, x_{\varphi^{\Right}(i)}) < \frac{\qq+1}{2}$.
	Hence, there is a unique shortest path from~$x_i$ to~$x_{\varphi^{\Left}(i)}$
	visiting~$x_{i+1 \mmod (\qq+1)}$,
	and there is a unique shortest path from~$x_i$ to~$x_{\varphi^{\Right}(i)}$
	visiting~$x_{i-1 \mmod (\qq+1)}$.
	This gives us that in the equation above, in the first sum, all shortest paths
	from~$x_i$ to~$t \in [x_{\varphi^{\Right}(i)}, x_\qq]$ visit~$x_0$, and in the second sum,
	all shortest paths from~$x_i$ to~$t \in [x_1, x_{\varphi^{\Left}(i)}]$ visit~$x_0$.
	If~$\varphi(x_i) \in \mathds{N}$, then there are two shortest paths from~$x_i$
	to~$x_{\varphi(i)}$, and one of them visits~$x_0$.
	With this we can rewrite the sum as follows:
	\begin{align*}
		\BC[x_0]
		&= \sum_{i=1}^{\varphi^{\Left}(0)}  \Big(
			\Pend[x_i] \cdot \Pend[x_{\varphi(i)}] \cdot \frac{1}{2}
			+ \sum_{t \in [x_{\varphi^{\Right}(i)}, x_a]} \Pend[x_i] \cdot \Pend[t]
			\Big)\\
		&+ \sum_{i=\varphi^{\Right}(0)}^{\qq} \Big(
			\Pend[x_i] \cdot \Pend[x_{\varphi(i)}] \cdot \frac{1}{2}
			+ \sum_{t \in [x_1, x_{\varphi^{\Left}(i)}]} \Pend[x_i] \cdot \Pend[t]
			\Big)\\
		&= \sum_{i=1}^{\varphi^{\Left}(0)} \Pend[x_i] \Big(
			\frac{1}{2} \Pend[x_{\varphi(i)}] + W[x_{\varphi^{\Right}(i)}, x_\qq] \Big)\\
		&+ \sum_{i=\varphi^{\Right}(0)}^{\qq} \Pend[x_i] \Big(
			\frac{1}{2} \Pend[x_{\varphi(i)}] + W[x_1, x_{\varphi^{\Left}(i)}] \Big).
	\end{align*}
	Since the~$\Pend[\cdot]$-values are given, the values~$\Wleft[\cdot]$ can be precomputed
	in~$O(\qq)$ time, and thus, when computing~$\BC[x_0]$,
	the values~$W[\cdot,\cdot]$ can be obtained in constant time.
	The values~$\varphi(i)$, $\varphi^{\Left}(i)$, and~$\varphi^{\Right}(i)$ can be computed in constant time as well, and thus it takes~$O(\qq)$ time to compute~$\BC[x_0]$.

	Assume now that we have computed~$\BC[x_k]$.
	Then we claim that for~$0 \le k < \qq$, $\BC[x_{k+1}]$ can be computed as follows:
	\begin{equation}
		\label{eq:bc-cpen-dp}
		\begin{aligned}
		\BC[x_{k+1}] = \BC[x_k]
		&- \Pend[x_{k+1}] \big(\Pend[x_{\varphi(k+1)}]\\
		&+ 2 W[x_{\varphi^{\Right}(k+1)},x_{k-1 \mmod (\qq+1)}] \big)\\
		&+ \Pend[x_k] \big(\Pend[x_{\varphi(k)}] +
			2 W[x_{k+2 \mmod (\qq+1)}, x_{\varphi^{\Left}(k)}]).
		\end{aligned}
	\end{equation}
	To this end, observe that all shortest paths in~$C$ that contain~$x_k$ as an inner vertex also contain~$x_{k+1}$ as an inner vertex, except for those paths that start or end in~$x_{k+1}$.
	Likewise, all shortest paths in~$C$ that contain~$x_{k+1}$ as an inner vertex also contain~$x_k$ as an inner vertex, except for those paths that start or end in~$x_k$.
	Hence, to compute~$\BC[x_{k+1}]$ from~$\BC[x_k]$, we need to subtract the~$\gamma$-values
	for shortest paths starting in~$x_{k+1}$ and visiting~$x_k$,
	and we need to add the~$\gamma$-values for shortest paths starting in~$x_k$ and
	visiting~$x_{k+1}$.
	Since by~\Cref{obs:gamma} each path contributes the same value to the \bc{} as its reverse, it holds
	\begin{align}
		\label{eq:bc-cpen-dp-gamma}
		\BC[x_{k+1}] &= \BC[x_k]
		+ 2 \cdot \sum_{t \in V(C)} \gamma(x_k, t, x_{k+1}) - \gamma(x_{k+1}, t, x_k).
	\end{align}
	With a similar argumentation as above for the computation of~$\BC[x_0]$ one can show
	that shortest paths starting in~$x_k$ and visiting~$x_{k+1}$ must end
	in~$t \in [x_{k+2}, x_{\varphi^{\Left}(k)}]$, or in~$x_{\varphi(k+1)}$.
	Shortest paths starting in~$x_{k+1}$ and visiting~$x_k$ must end
	in~$t \in [x_{\varphi^{\Right}(k+1)},x_{k-1}]$, or in~$x_{\varphi(k)}$.
	Just as above, for both~$i = k$ and~$i = k+1$, some fixed vertex~$x_j$ is
	visited by only half of the shortest paths from~$x_i$ to~$x_{\varphi(i)}$.
	With the arguments above we can rewrite~\Cref{eq:bc-cpen-dp-gamma} to obtain
	the claimed~\Cref{eq:bc-cpen-dp}.

	After precomputing the values~$\Wleft[\cdot]$ and~$\BC[x_0]$ in~$O(\qq)$ time and space, we can compute each of the values~$\BC[x_{k+1}]$ for~$0 \le k < \qq$ in constant time.
	Hence, the procedure requires~$O(\qq)$ time.
\end{proof}

\subsection{Dealing with Other Biconnected Graphs}
Recall that, after our preprocessing on all cut vertices and degree-one vertices, we obtain a graph consisting of biconnected components, each of which can be solved independently.
Also, in the previous subsection, we showed how to solve \textsc{Weighted Betweenness Centrality} on cycles.
It remains to show how to solve the problem on biconnected graphs that are not cycles (but may contain one).

\paragraph{Remark.}
Henceforth, in this paper, we assume that we are given a vertex-weighted biconnected graph that is not a cycle.

\paragraph{Outline of the algorithmic approach.}
Starting with a vertex-weighted biconnected graph, our algorithm focuses on degree-two vertices.
In contrast to degree-one vertices, degree-two vertices can lie on shortest paths between two other vertices.
Moreover, different degree-two vertices on the same shortest path can have different \bc{} values (see \cref{fig:bc-hard} for an example). 
This makes degree-two vertices harder to handle: 
Removing a degree-two vertex~$v$ in a similar way as done with degree-one vertices (see \Cref{rr:deg-one-vertices}) potentially affects many other shortest paths that neither start nor end in~$v$.
Thus, we treat degree-two vertices differently:
Instead of removing vertices one-by-one, we process multiple degree-two vertices at once and exploit that consecutive degree-two vertices share many shortest paths they lie on.
To this end we introduce the notion of \maxpaths{}.

\begin{figure}
	\centering
	\begin{tikzpicture}[scale=1.9]
		\usetikzlibrary{shapes}
		\tikzstyle{pnode}=[fill,circle,scale=3/8]
		\tikzstyle{d1}=[red,very thick]
		\tikzstyle{d2}=[blue,dotted,very thick]
		\tikzstyle{vertex}=[draw,circle,inner sep=1pt,minimum size=7pt,fill=white]
		\tikzstyle{edge}=[thick,color=black!50]
		\def\centerarc[#1](#2)(#3:#4:#5) { \draw[#1] ($(#2)+({#5*cos(#3)},{#5*sin(#3)})$) arc (#3:#4:#5); };

		\foreach[count=\i] \name / \x / \y / \z in {{x_0}/180/14.5/180, {x_1}/216/13.0/180, {x_2}/252/18.0/270, {x_3}/288/23.0/270, {x_4}/324/28.0/0, {x_5}/0/44.5/0} {
			\pgfmathtruncatemacro{\j}{\i - 1}
			\node (x\j) at (\x:1) [vertex,label=\z:{\small$\y$}] {$\name$};
		}

		\node (a1) at (144:0.7) [vertex] {$a_1$};
		\node (a2) at (144:1.3) [vertex] {$a_2$};
		\node (b) at (108:1) [vertex] {};
		\node (q) at (90:1) [vertex] {};
		\node (c) at (72:1) [vertex] {};
		\node (d1) at (36:0.6) [vertex] {$d_1$};
		\node (dots) at (36:1) [rotate=36,align=center] {$\dots$};
		\node (dj) at (36:1.4) [vertex] {$d_7$};

		\draw[edge] (x0) -- (x1);
		\draw[edge] (x1) -- (x2);
		\draw[edge] (x2) -- (x3);
		\draw[edge] (x3) -- (x4);
		\draw[edge] (x4) -- (x5);

		\draw[edge] (x0) -- (a1);
		\draw[edge] (x0) -- (a2);
		\draw[edge] (a1) -- (b);
		\draw[edge] (a2) -- (b);
		\draw[edge] (b) -- (q);
		\draw[edge] (q) -- (c);
		\draw[edge] (c) -- (d1);
		\draw[edge] (c) -- (dj);
		\draw[edge] (x5) -- (d1);
		\draw[edge] (x5) -- (dj);

		\tikzstyle{path} = [color=black,opacity=.15,line cap=round, line join=round, line width=8pt]
		\begin{pgfonlayer}{background}
			\draw[path] (x3.center) -- (x2.center) -- (x1.center) -- (x0.center) -- (a1.center);
			\draw[path] (x3.center) -- (x2.center) -- (x1.center) -- (x0.center) -- (a2.center);
		\end{pgfonlayer}
		
		\begin{scope}[xshift=3.25cm]
			\foreach[count=\i] \name / \x / \y / \z in {{x_0}/180/14.5/180, {x_1}/216/13.0/180, {x_2}/252/18.0/270, {x_3}/288/23.0/270, {x_4}/324/28.0/0, {x_5}/0/44.5/0} {
				\pgfmathtruncatemacro{\j}{\i - 1}
				\node (x\j) at (\x:1) [vertex,label=\z:{\small$\y$}] {$\name$};
			}

			\node (a1) at (144:0.7) [vertex] {$a_1$};
			\node (a2) at (144:1.3) [vertex] {$a_2$};
			\node (b) at (108:1) [vertex] {};
			\node (q) at (90:1) [vertex] {};
			\node (c) at (72:1) [vertex] {};
			\node (d1) at (36:0.6) [vertex] {$d_1$};
			\node (dots) at (36:1) [rotate=36,align=center] {$\dots$};
			\node (dj) at (36:1.4) [vertex] {$d_7$};

			\draw[edge] (x0) -- (x1);
			\draw[edge] (x1) -- (x2);
			\draw[edge] (x2) -- (x3);
			\draw[edge] (x3) -- (x4);
			\draw[edge] (x4) -- (x5);

			\draw[edge] (x0) -- (a1);
			\draw[edge] (x0) -- (a2);
			\draw[edge] (a1) -- (b);
			\draw[edge] (a2) -- (b);
			\draw[edge] (b) -- (q);
			\draw[edge] (q) -- (c);
			\draw[edge] (c) -- (d1);
			\draw[edge] (c) -- (dj);
			\draw[edge] (x5) -- (d1);
			\draw[edge] (x5) -- (dj);

			\begin{pgfonlayer}{background}
				\draw[path] (x2.center) -- (x3.center) -- (x4.center) -- (x5.center) -- (d1.center);
				\draw[path] (x2.center) -- (x3.center) -- (x4.center) -- (x5.center) -- (dots);
				\draw[path] (x2.center) -- (x3.center) -- (x4.center) -- (x5.center) -- (dj.center);
			\end{pgfonlayer}
		\end{scope}

	\end{tikzpicture}
	\caption{
		\label{fig:bc-hard}
		An example graph containing a \maxpath{} $x_0 \dots x_5$ (see \cref{def:maxpath}).
		The labels give the \bc{} values of the vertices.
		Marked are shortest paths from~$a_i$ to~$x_3$ (left-hand side) and from~$d_i$ to~$x_2$ (right-hand side).
		The former affect the \bc{} value of~$x_2$, but not of~$x_3$; the latter affect the \bc{} value of~$x_3$, but not of~$x_2$.
		Hence, most, but not all, of the paths traversing through~$x_2$ also affect the \bc{} value of~$x_3$.
		Note that this difference cannot be decided locally within the \maxpath{}, but can have an arbitrary effect on the difference arbitrarily far away in the graph.
		In this example graph, one could add more and more ``$d$-vertices'' (the figure shows~$d_1$--$d_7$) to further increase the difference in the \bc{} values of~$x_2$ and~$x_3$.
	}
\end{figure}

\begin{definition}
	\label[definition]{def:maxpath}
	\emph{
	Let~$G$ be a graph. 
	A path~$P = v_0\ldots v_\ell$ is a \emph{\maxpath{}} in~$G$ if $\ell \ge 2$ and the inner vertices~$v_1, \ldots, v_{\ell-1}$ all have degree two in~$G$, but the endpoints~$v_0$ and~$v_\ell$ do not, that is, $\deg_G(v_1) = \ldots = \deg_G(v_{\ell-1}) = 2$, $\deg_G(v_0) \neq 2$, and $\deg_G(v_\ell) \neq 2$.
Moreover, $\AllPmax$ is the set of all \maxpaths{} in~$G$.}
\end{definition}

%
In a nutshell, our algorithm treats each biconnected component of the input graph in the following three stages (compare with \Cref{alg:overview}):

\begin{algorithm}[p]
	\caption{Algorithm for computing  betweenness centrality of a biconnected graph that is not a cycle.} 
	\label{alg:overview}
	\KwIn{An undirected biconnected graph~$G$ with vertex weights~$\Pend \colon V(G) \rightarrow \N$.}
	\KwOut{The betweenness centrality values of all vertices.}
 	\vspace{1em}
	
	\lForEach{$v \in V(G)$}{$\BC[v] \gets 0$\;
	\tcp*[f]{\it $\BC$ will contain the betweenness centrality values}}
		$\AllPmax \gets{}$all \maxpaths{} of~$G$ \label{line:compute-max-paths}\;
		\tcp*{\it computable in $O(n+m)$ time, see \Cref{lem:compute-maximal-pendcycles}}
		
		\ForEach(\label{line:init-tables}){$s \in \Vthree(G)$}
		{
			\tcp*{\it some precomputations taking~$O(kn)$ time, see \Cref{lem:init-tables}}
			compute~$d_G(s,t)$ and~$\sigma_{st}$ for each~$t \in V(G) \setminus \{s\}$\;
			$\Inc[s,t] \gets 2 \cdot \Pend[s]\cdot\Pend[t] / \sigma_{st}$ for each~$t \in \Vtwo(G)$ \label{line:init-inc-table-vtwo} \;
			$\Inc[s,t] \gets \Pend[s]\cdot\Pend[t] / \sigma_{st}$ for each~$t \in \Vthree(G) \setminus \{s\}$ \label{line:init-inc-table}
		}
		\ForEach( \label{line:wleftright} ){$x_0x_1\ldots x_\qq = \Pmax \in \AllPmax$}
		{
			\tcp*{\it initialize~$\Wleft$ and~$\Wright$ in~$O(n)$ time}
			$\Wleft[x_0] \gets \Pend[x_0]$; $\Wright[x_\qq] \gets \Pend[x_\qq]$\;
			\lFor{$i = 1$ \KwTo $\qq$}
			{
				$\Wleft[x_i] \gets \Wleft[x_{i-1}] + \Pend[x_i]$
			}
			\lFor{$i = \qq-1$ \KwTo $0$}
			{
				$\Wright[x_i] \gets \Wright[x_{i+1}] + \Pend[x_i]$ \label{line:compute-wright}
			}
		
		\ForEach( \label{line:maxpaths-loop}){$x_0x_1\ldots x_\qq = \Pmax_1 \in \AllPmax$}
		{
			\tcp*[f]{\it case~$s \in \Vtwo(\Pmax_1)$, see \Cref{sec:maxpaths}}\\
			\tcc{\it deal with the case~$t \in \Vtwo(\Pmax_2)$, see \Cref{sec:two-maxpaths}}
			\ForEach{$y_0y_1\ldots y_\qr = \Pmax_2 \in \AllPmax \setminus \{\Pmax_1\}$}
			{
				\tcc{\it update~$\BC$ for the case~$v \in V(\Pmax_1) \cup V(\Pmax_2)$}
				\lForEach{$v \in V(\Pmax_1) \cup V(\Pmax_2)$}
				{
					$\BC[v] \gets \BC[v] + \gamma(s,t,v)$
				}
				\tcc{\it now deal with the case~$v \notin V(\Pmax_1) \cup V(\Pmax_2)$}
				update~$\Inc[x_0,y_0]$, $\Inc[x_\qq,y_0]$, $\Inc[x_0,y_\qr]$, and~$\Inc[x_\qq,y_\qr]$ \label{line:inc-update-two-paths}
			}
			
			\tcc{\it deal with the case that~$t \in \Vtwo(\Pmax_1)$, see Section~\ref{sec:one-maxpath}}
			\lForEach{$v \in V(\Pmax_1)$}
			{
				$\BC[v] \gets \BC[v] + \gamma(s,t,v)$
			}
			update~$\Inc[x_0,x_\qq]$ \tcp*{\it this deals with the case~$v \notin V(\Pmax_1)$} \label{line:inc-update-one-path}
		}
		
		\ForEach{$s \in \Vthree(G)$}
		{
			\tcp*{\it perform modified BFS from~$s$, see \Cref{sec:postproc-summary}}\label{line:modified-bfs-loop}
			\lForEach{$t,v \in V(G)$}
			{
				$\BC[v] \gets \BC[v] + \Inc[s,t] \cdot \sigma_{st}(v)$ \label{line:modified-bfs}
			}
		}
	}
	\KwRet{$\BC$}.
\end{algorithm}

\begin{enumerate}
	\item For all pairs~$s, t$ of vertices of degree at least three, precompute~$d_G(s, t)$ and~$\sigma_{s t}$, and initialize~$\Inc[s, t]$ (see \Crefrange{line:init-tables}{line:init-inc-table}).
	\item Compute \bc{} values for paths starting and ending  in \maxpaths{} and store them in~$\Inc[\cdot,\cdot]$, considering two cases (see \Crefrange{line:maxpaths-loop}{line:inc-update-one-path}):
		\begin{compactitem}[--]
	\item both endpoints of the path are in the same \maxpath{};
			\item the endpoints are in two different \maxpaths{}.
		\end{compactitem}
	\item In a postprocessing step, compute the \bc{} for all remaining paths (at least one endpoint is of degree at least three) and incorporate the values stored in~$\Inc[\cdot,\cdot]$ (see \Crefrange{line:modified-bfs-loop}{line:modified-bfs}).
\end{enumerate}

Note that in a biconnected graph that is not a cycle, every degree-two vertex is an inner vertex of a \maxpath{}.
If some degree-two vertex~$v$ was not contained in a \maxpath{}, then~$v$ would be contained in a cycle that contains exactly one vertex~$u$ that is of degree at least three.
But then~$u$ is a cut vertex and the graph would not be biconnected; a contradiction.
The remaining part of the algorithm deals with \maxpaths{}.
Note that if the (biconnected) graph is not a cycle, then all degree-two vertices are contained in maximal induced paths:

Using standard arguments, we can show that the number of \maxpaths{} is upper-bounded by the minimum of the feedback edge number~$k$ of the input graph and the number~$n$ of vertices.
Moreover, one can easily compute all \maxpaths{} in linear-time (see \Cref{line:compute-max-paths} of \cref{alg:overview}).

\begin{lemma}
	\label[lemma]{lem:fen}
	Let~$G$ be a graph with feedback edge number~$k$ that does not contain degree-one vertices.
	Then~$G$ contains at most~$\min \{n, 2k\}$ vertices of degree at least three and at most~$\min \{n, 3k\}$ \maxpaths{}.
\end{lemma}
\begin{proof}
	Recall that our graph is biconnected.
	Thus~$\sum_{v \in V(G)} \deg(v) = 2m = 2(n-1+k)$, and 
	\begin{align*}
		2 (n-1) + k &= 2(|\Vtwo(G)| + |\Vthree(G)| - 1 + k)\\
			    &= \sum_{v \in V(G)} \deg(v)
			    = \sum_{v \in \Vtwo(G)} \deg(v) + \sum_{v \in \Vthree(G)} \deg(v)\\
			    &\ge 2 \cdot |\Vtwo(G)| + 3 \cdot |\Vthree(G)|.
	\end{align*}
	Solving for~$|\Vthree(G)|$ gives us that there are at most~$2k - 2$ vertices of degree at least three.
	Then~$\sum_{v \in \Vthree(G)} \deg(v) = 3 |\Vthree(G)| \le 6k - 6$.
	It follows that there are at most~$3k$ paths whose endpoints are in~$\Vthree(G)$, hence~$|\AllPmax| \le 3k -3$.
	Clearly, for both the number of vertices of degree at least three and number of \maxpaths{}, $n$ is also a valid upper bound.
%
\end{proof}

\begin{lemma}
	\label[lemma]{lem:compute-maximal-pendcycles}
	The set~$\AllPmax$ of all \maxpaths{} of a graph with~$n$ vertices and~$m$ edges can be computed in~$O(n+m)$ time.
\end{lemma}
\begin{proof}
	Iterate through all vertices~$v \in V(G)$.
	If~$v \in \Vtwo(G)$, then iteratively traverse the two edges incident to~$v$ to discover adjacent degree-two vertices until finding endpoints~$v_\ell, v_r \in \Vthree(G)$.
	If~$v_\ell = v_r$, then we found a cycle which can be ignored.
	Otherwise, we have a \maxpath{}~$\Pmax = v_\ell \dots v_r$, which we add to~$\AllPmax$.
	
	Note that every degree-two vertex is contained either in exactly one \maxpath{} or in exactly one cycle.
	Hence, we do not need to reconsider any degree-two vertex found in the traversal above and we can find all \maxpaths{} in~$O(m+n)$ time. 
\end{proof}

Our algorithm processes the \maxpaths{} one by one (see \Crefrange{line:init-tables}{line:modified-bfs}).
This part of the algorithm requires pre- and postprocessing (see \Crefrange{line:init-tables}{line:compute-wright} and \Crefrange{line:modified-bfs-loop}{line:modified-bfs} respectively).
In the preprocessing, we initialize tables that are frequently used in the main part (of \Cref{sec:maxpaths}).
The postprocessing computes the final betweenness centrality values of each vertex as this computation is too time-consuming to be executed for each \maxpath{}.
When explaining our basic ideas, we will first present the postprocessing as this explains why certain values will be computed during the algorithm.

Recall that we want to compute~$\sum_{s,t \in V(G)} \gamma(s,t,v)$ for each~$v\in V(G)$ (see \Cref{eq:betw-cent-with-pend}).
Using the following observations, we split~\Cref{eq:betw-cent-with-pend} into different parts.
\begin{observation}
	\label[observation]{obs:gamma}
	For~$s, t, v \in V(G)$ it holds that~$\gamma(s, t, v) = \gamma(t, s, v)$.
\end{observation}

\begin{observation}
	\label[observation]{obs:split-bc-sum}
	Let~$G$ be a biconnected graph with at least one vertex of degree at least three.
	Let~$v \in V(G)$. Then, 
	\begin{align*}
		\sum_{s,t \in V(G)}
		\gamma(s, t, v)
		&= \sum_{s \in \Vthree(G),\, t \in V(G)}
		\gamma(s, t, v)
		+ \sum_{s \in \Vtwo(G),\, t \in \Vthree(G)}
		\gamma(t, s, v)\\&
		+ \sum_{\substack{
			s \in \Vtwo(\Pmax_1),\, t \in \Vtwo(\Pmax_2)\\
			\Pmax_1 \ne \Pmax_2 \in \AllPmax}}
		\gamma(s, t, v)
		+ \sum_{\substack{
			s, t \in \Vtwo(\Pmax)\\
			\Pmax \in \AllPmax}}
		\gamma(s, t, v).
	\end{align*}
\end{observation}
\begin{proof}
	The first two sums cover all pairs of vertices in which at least one of the two vertices
	is of degree at least three.
	The other two sums cover all pairs of vertices which both have degree two.
	As all vertices of degree two must be part of some \maxpath{}, we have~$\Vtwo(G) = \Vtwo(\bigcup \AllPmax)$.
	Two vertices of degree two can thus either be in two different \maxpaths{} (third sum)
	or in the same \maxpath{} (fourth sum).
\end{proof}

In the remaining graph, by \Cref{lem:fen}, there are at most~$O(\min\{k,n\})$ vertices of degree at least three and at most~$O(k)$ \maxpaths.
This implies that we can afford to run the modified BFS (similar to Brandes' algorithm) from each vertex~$s \in \Vthree(G)$ in~$O(\min\{k,n\} \cdot (n+k)) = O(kn)$ time.
This computes the first summand and, by \Cref{obs:gamma}, also the second summand in \Cref{obs:split-bc-sum}.
However, we cannot afford to run such a BFS from every vertex. 
Thus, we need to compute the third and fourth summand differently.

To this end, note that~$\sigma_{st}(v)$ is the only term in~$\gamma(s,t,v)$ that depends on~$v$.
Our goal is to precompute~$\gamma(s,t,v) / \sigma_{st}(v) = \Pend[s] \cdot \Pend[t] / \sigma_{st}$ for as many vertices as possible.
Hence, we store precomputed values in a table~$\Inc[\cdot,\cdot]$ (see \Cref{line:init-inc-table,line:inc-update-two-paths,line:inc-update-one-path}).
Then, we plug this factor into the next lemma which provides our postprocessing. 
\begin{lemma}
	\label[lemma]{lem:brandesalg}
	Let~$s$ be a vertex and let $f \colon V(G)^2 \to \mathds{N}$ be a function such that for each~$u,v\in V(G)$ the value~$f(u,v)$ can be computed in~$O(\tau)$ time.
	Then, for all~$v \in V(G)$ one can compute the value~${\sum_{t \in V(G)} f(s,t) \cdot \sigma_{st}(v)}$ in~$O(n \cdot \tau + m)$ time.
\end{lemma}
\begin{proof}
	This proof generally follows the structure of the proof by \citet[Theorem~6, Corollary~7]{Bra01}, the main difference being the generalization of the distance function to an arbitrary function~$f$.

	Analogously to Brandes we define~$\sigma_{st}(v,w)$ as the number of shortest paths from~$s$ to~$t$ that contain the edge~$\{v,w\}$, and~$S_s(v)$ as the set of successors of a vertex~$v$ on shortest paths from~$s$, that is, $S_s(v) = \{w\in V(G) \mid \{v,w\}\in E \land d_G(s,w) = d_G(s,v) + 1\}$. 
	For the sake of readability we also define~$\chi_{sv} = \sum_{t\in V(G)} f(s,t) \cdot \sigma_{st}(v)$.
	We will first derive a series of equations that show how to compute~$\chi_{sv}$.
	Afterwards we justify \Cref{eq:chi1,eq:chi2}.
	\begin{align*}
	\chi_{sv} =& \sum_{t \in V(G)} f(s,t) \cdot \sigma_{st}(v)\\
	= &\sum_{t \in V(G)}  f(s,t) \sum_{w \in S_s(v)} \sigma_{st}(v,w)=\sum_{w \in S_s(v)} \sum_{t \in V(G)}  f(s,t) \cdot  \sigma_{st}(v,w)\stepcounter{equation}\tag{\theequation}\label{eq:chi1}\\
	\phantom{\chi_{sv}}= &\sum_{w \in S_s(v)} \Big (  \big( \sum_{t \in V(G)\setminus \{w\}}  f(s,t) \cdot \sigma_{st}(v,w) \big) + f(s,w)  \cdot \sigma_{sw}(v,w) \Big )\\
	= &\sum_{w \in S_s(v)} \Big (  \big( \sum_{t \in V(G)\setminus \{w\}}  f(s,t) \cdot \sigma_{st}(w) \cdot \frac{\sigma_{sv}}{\sigma_{sw}} \big) + f(s,w)  \cdot \sigma_{sv} \Big )\stepcounter{equation}\tag{\theequation}\label{eq:chi2}\\
	= &\sum_{w \in S_s(v)} \Big (  \chi_{sw} \cdot \frac{\sigma_{sv}}{\sigma_{sw}}  + f(s,w)  \cdot \sigma_{sv} \Big )
	\end{align*}
	We will now show that \Cref{eq:chi1,eq:chi2} are correct.
	All other equalities are based on simple arithmetics.
	To see that \Cref{eq:chi1} is correct, observe that each shortest path from~$s$ to any other vertex~$t$ that contains~$v$ either ends in~$v$, that is,~$t = v$, or contains exactly one edge~$\{v,w\}$, where~${w \in S_s(v)}$.
	If~$t = v$, then~$\sigma_{st}(v) = 0$ and therefore~$\sum_{t \in V} \sigma_{st}(v) =  \sum_{t \in V}\sum_{w \in S_s(v)} \sigma_{st}(v,w)$.
	To see that \Cref{eq:chi2} is correct, observe the following: 
	First, note that the number of shortest paths from~$s$ to~$t$ that contain a vertex~$v$ is \[\sigma_{st}(v) = \begin{cases} 0, & \text{if }d_G(s,v) + d_G(v,t) > d_G(s,t);\\\sigma_{sv} \cdot \sigma_{vt}, & \text{otherwise;}\end{cases}\]
	second, note that the number of shortest $st$-paths that contain an edge~$\{v,w\}$,~$w \in S_s(v)$, is \[\sigma_{st}(v,w) = \begin{cases} 0, & \text{if }d_G(s,v) + d_G(w,t) + 1 > d_G(s,t);\\\sigma_{sv} \cdot \sigma_{wt}, & \text{otherwise;}\end{cases}\]
	and third, note that the number of shortest~$sw$-paths that contain~$v$ is equal to the number of shortest~$sv$-paths.
	The combination of these three observations yields~$\sigma_{st}(v,w) =  \sigma_{sv} \cdot \sigma_{wt} = \sigma_{sv} \cdot \sigma_{st}(w) / \sigma_{sw}.$

	We next show how to compute~$\chi_{sv}$ for all~$v\in V$ in~$O(m + n \cdot \tau)$ time.
	First, order the vertices in non-increasing distance to~$s$ and compute the set of all successors of each vertex in~$O(m)$ time using breadth-first search.
	Note that the number of successors of all vertices is at most~$m$ since each edge defines at most one successor-predecessor relation.
	Then compute~$\chi_{sv}$ for each vertex by a dynamic program that iterates over the ordered list of vertices and computes
	\[\sum_{w \in S_s(v)} \Big (  \chi_{sw} \cdot \frac{\sigma_{sv}}{\sigma_{sw}}  + f(s,w)  \cdot \sigma_{sv} \Big )\]
	in overall~$O(m + n \cdot \tau)$ time.
	This can be done by first computing~$\sigma_{st}$ for all~$t\in V$ in overall~$O(m)$ time due to \citet[Corollary 4]{Bra01} and~$f(s,t)$ for all~$t \in V(G)$ in~$O(n \cdot \tau)$ time, and then using the already computed values~$S_s(v)$ and~$\chi_{sw}$ to compute
	\[\chi_{sv} = \sum_{w \in S_s(v)} \Big (  \chi_{sw} \cdot \frac{\sigma_{sv}}{\sigma_{sw}}  + f(s,w)  \cdot \sigma_{sv} \Big )\]
	in~$O(|S_s(v)|)$ time.
	Note that~$\sum_{v\in V} |S_s(v)| \leq O(m)$.
	This concludes the proof.
\end{proof}

The proof of \cref{lem:brandesalg} provides us with an algorithm.
Our goal is then to only start this algorithm from few vertices, specifically the vertices of degree at least three (see \Cref{line:modified-bfs} of \Cref{alg:overview}).
Since the term~$\tau$ in the above lemma will be constant, we obtain a running time of~$O(kn)$ for running this postprocessing on all vertices of degree at least three.
The most intricate part will be to precompute the factors in $\Inc[\cdot,\cdot]$ (see \Cref{line:inc-update-two-paths,line:inc-update-one-path} of \Cref{alg:overview}).
We defer the details to \Cref{sec:one-maxpath,sec:two-maxpaths}.
In these parts, we need the tables~$\Wleft$ and~$\Wright$.
These tables store values depending on the \maxpath{} a vertex is in.
More precisely, for a vertex~$x_i$ in a \maxpath{}~$\Pmax = x_0x_1\ldots x_\qq$, we store in~$\Wleft[x_k]$ the sum of the~$\Pend[\cdot]$\nobreakdash-values of vertices ``left of''~$x_k$ in~$\Pmax$; formally, $\Wleft[x_k] = \sum_{i=1}^{k} \Pend[x_i]$.
Similarly, we have~$\Wright[x_k] = \sum_{i=k}^{\qq-1} \Pend[x_i]$.
The reason for having these tables is easy to see:
Assume for the vertex~$x_k \in \Pmax$ that the shortest paths to~$t \notin V(\Pmax)$ leave~$\Pmax$ through~$x_0$.
Then, it is equivalent to just consider the shortest path(s) starting in~$x_0$ and simulate the vertices between~$x_k$ and~$x_0$ in~$\Pmax$ by ``temporarily increasing''~$\Pend[x_0]$ by~$\Wleft[x_k]$.
This is also the idea behind the argument that we only need to increase the values~$\Inc[\cdot,\cdot]$ for the endpoints of the \maxpaths{} in \Cref{line:inc-update-two-paths} of \cref{alg:overview}.

This leaves us with the remaining part of the preprocessing: the computation of the distances~$d_G(s,t)$, the number of shortest paths~$\sigma_{st}$, and~$\Inc[s,t]$ for~$s \in \Vthree(G), t \in V(G)$ (see \Crefrange{line:init-tables}{line:init-inc-table}).
This can be done in~$O(kn)$ time as well:
\begin{lemma}
	\label[lemma]{lem:init-tables}
	The initialization in the for-loop in \Crefrange{line:init-tables}{line:init-inc-table} of \Cref{alg:overview} can be done in~$O(kn)$ time.
\end{lemma}
\begin{proof}
	Following \citet[Corollary 4]{Bra01}, computing the distances and the number of shortest paths
	from a fixed vertex~$s$ to every~$t \in V(G)$ takes~$O(m) = O(n+k)$ time.
	Once these values are computed for a fixed~$s$, computing~$\Inc[s, t]$ for~$t \in V(G)$ takes~$O(n)$ time since the values~$\Pend[s]$,~$\Pend[t]$, and~$\sigma_{s t}$ are known.
	Since, by \Cref{lem:fen}, there are~$O(\min \{k, n\})$ vertices of degree at least three, it takes~$O(\min\{k, n\} \cdot (n+k + n)) = O(kn)$ time to compute \Crefrange{line:init-tables}{line:init-inc-table}.
\end{proof}
%
%

\section{Dealing with \maxpaths{}}
\label{sec:maxpaths}

In this section, we focus on degree-two vertices contained in \maxpaths. Recall that the goal is to compute the betweenness centrality~$C_B(v)$ (see \Cref{eq:betw-cent-with-pend}) for all~$v \in V(G)$ in~$O(kn)$ time.
In the end of this section, we finally prove our main theorem (\Cref{thm:main}).

\Cref{fig:diagram} shows the general proof structure of the main theorem.
\begin{figure}
	\centering
	\small
	\begin{tikzpicture}[
		every node/.style = {draw, align=center},
		sibling distance=17em,
		every child/.style={edge from parent/.style={draw, {-{Latex}}}},
	]
		\node {\cref{thm:main}}
			child[edge from parent fork down, level distance=6.0em, sibling distance=13.5em] { node {both endpoints in\\different paths\\ (\cref{prop:mp-pairs})}
				child[sibling distance=6.5em] { node {$v$ outside\\ of the paths\\(\cref{lem:bc-two-mp-outside})} }
				child[sibling distance=6.5em] { node {$v$ inside\\ one path\\ (\cref{lem:bc-two-mp-inside}$^*$)}
					child[level distance=4.0em] { node {symmetry (\cref{lem:bc-two-mp-inside-two}$^*$)} }
				}
			}
			child[edge from parent fork down, level distance=6.0em, sibling distance=13.5em] { node {both endpoints in\\the same path\\ (\cref{prop:mp-single})}
				child[sibling distance=6.5em] { node {$v$ inside\\the path\\(\cref{lem:v-in-one-maxpath}$^*$)} }
				child[sibling distance=6.5em] { node (post2) {$v$ outside\\of the path\\(\cref{lem:v-not-within-maxpath}$^*$)} 
					child[level distance=4.0em,xshift=5em] { node (post) {postprocessing (\cref{lem:brandesalg})} edge from parent[draw=none] }
				}
			}
			child[edge from parent fork down, level distance=6.0em, sibling distance=13.5em] { node (post1) {at least one end-\\ point of degree\\ at least three} };

		\draw[-{Latex}] (post1) to (post);
		\draw[-{Latex}] (post2) to (post);

	\end{tikzpicture}
	\caption{
		Structure of how the proof of \cref{thm:main} is split into different cases.
		By ``paths'' we mean \maxpaths{}.
		The first layer below the main theorem specifies the positions of the endpoints~$s$ and~$t$, whereas the second layer specifies the position of the vertex~$v$, for which the \bc{} is computed.
		The third layer displays further lemmata used to prove the corresponding lemma above.
		Proofs of lemmata marked with an asterisk are deferred to the appendix.
	}
	\label{fig:diagram}
\end{figure}
Based on \Cref{obs:split-bc-sum}, which we use to split the sum in \Cref{eq:betw-cent-with-pend} in the definition of \textsc{Weighted Betweenness Centrality}, we compute~$C_B(v)$ in three steps.
By starting a modified BFS from vertices in~$\Vthree(G)$ similarly to \citet{BGPL12} and \citet{Bra01}, we can compute 
\begin{align*}
	\sum_{s \in \Vthree(G), t \in V(G)} \gamma(t, s, v) + \sum_{s \in \Vtwo(G), t \in \Vthree(G)} \gamma(s, t, v)
\end{align*}
for all~$v \in V(G)$ in overall~$O(kn)$ time. 
In the next two subsections, we show how to compute the remaining two summands given in \Cref{obs:split-bc-sum} (i.e., we prove \cref{prop:mp-pairs,prop:mp-single}).
In the last subsection, we prove \Cref{thm:main}.

\subsection{Paths with endpoints in different \maxpaths{}}
\label{sec:two-maxpaths}
In this subsection, we look at shortest paths between pairs of \maxpaths{}~$\Pmax_1 = x_0 \dots x_\qq$ and~$\Pmax_2 = y_0 \dots y_\qr$, and how to efficiently determine how these paths affect the \bc{} of each vertex.
\newcommand{\propmppairs}{%
	In~$O(kn)$ time one can compute the following values for all~$v \in V(G)$:
	\begin{equation*}
		\sum_{\substack{
			s \in \Vtwo(\Pmax_1),\, t \in \Vtwo(\Pmax_2)\\
			\Pmax_1 \ne \Pmax_2 \in \AllPmax}}
			\gamma(s, t, v).
	\end{equation*}
}
\begin{proposition}
	\label[proposition]{prop:mp-pairs}
	\propmppairs
\end{proposition}

In the proof of \Cref{prop:mp-pairs}, we consider two cases for every pair~$\Pmax_1 \ne \Pmax_2 \in \AllPmax$ of \maxpaths{}:
First, we look at how the shortest paths between vertices in~$\Pmax_1$ and~$\Pmax_2$
affect the \bc{} of those vertices that are not contained in the two \maxpaths{},
and second, how they affect the \bc{} of those vertices that are contained in the two \maxpaths{}.
Finally, we prove \Cref{prop:mp-pairs}.

Throughout the following proofs, we will need the following definitions (see \cref{fig:twopaths} for an illustration).
\begin{figure}
	\centering
	\begin{tikzpicture}[rotate=90]
		\usetikzlibrary{calc}
		\tikzstyle{pnode}=[fill,circle,scale=1/3]
		\tikzstyle{d1}=[red,very thick]
		\tikzstyle{d2}=[blue,dashed,very thick]
		\def\centerarc[#1](#2)(#3:#4:#5) { \draw[#1] ($(#2)+({#5*cos(#3)},{#5*sin(#3)})$) arc (#3:#4:#5); };

		\node (x0) at (135:2 cm) [pnode,label=180:{$x_0$}]{};
		\node (xq) at (45:2 cm) [pnode,label=180:{$x_\qq$}]{};


		\node (xtmid) at (90:2 cm) [pnode,label=180:{$\xmid_t$}]{};
		\node (xtleft) at (70:2 cm) [pnode,label=180:{$\xright_t$}]{};
		\node (xtright) at (110:2 cm) [pnode,label=180:{$\xleft_t$}]{};
		\node (t) at (270:2 cm) [pnode,label=0:{$t \in \Pmax_2$}]{};

		\node (p1max) at (0,4) {$\Pmax_1$};

		\centerarc[black](0,0)(70:90:2)
		\centerarc[black](0,0)(90:110:2)

		\centerarc[black,densely dotted](0,0)(45:70:2)
		\centerarc[black,densely dotted](0,0)(110:135:2)

		\centerarc[black,decorate,decoration=snake](0,0)(135:269:2)
		\centerarc[black,decorate,decoration=snake](0,0)(-90:45:2)

		\centerarc[d1](0,0)(110:270:1.7)
		\centerarc[d2](0,0)(-90:70:1.7)

	\end{tikzpicture}
	\caption{
		An exemplary graph containing two \maxpaths{}~$\Pmax_1 = x_0 \ldots x_\qq$ and~$\Pmax_2$.
		The curled lines depict shortest paths from~$t$ to~$x_0$ and to~$x_\qq$ respectively.
		We then choose~$\xleft_t, \xmid_t, \xright_t \in V(\Pmax_1)$ in such a way that the distance from~$t$ to~$\xleft_t$ and to~$\xright_t$ is equal, that is, the red (solid) line and the blue (dashed) line represent shortest paths of same length.
		Since~$\xmid_t$ is adjacent to~$\xleft_t$ and~$\xright_t$, there are shortest paths from~$\xmid_t$ to~$t$ via both~$x_0$ and~$x_\qq$, that is, along the blue and the red line.
	}
	\label{fig:twopaths}
\end{figure}
Let~$t \in \Pmax_2$. 
Then we choose vertices~$\xleft_t, \xright_t \in \Vtwo(\Pmax_1)$ such that shortest paths from~$t$ to~$s \in \{x_1, x_2, \dots, \xleft_t\} \eqqcolon \Xleft_t$ enter~$\Pmax_1$ only via~$x_0$, and shortest paths from~$t$ to~$s \in \{\xright_t, \dots, x_{\qq-2}, x_{\qq-1}\} \eqqcolon \Xright_t$ enter~$\Pmax_1$ only via~$x_\qq$.
There may exist a vertex~$\xmid_t$ to which there are shortest paths both via~$x_0$ and via~$x_\qq$.
For computing the indices of these vertices, we determine an index~$i$ such that~$d_G(x_0, t) + i = d_G(x_\qq, t) + \qq - i$ which is equivalent to~$i = \frac{1}{2}(\qq - d_G(x_0, t) + d_G(x_\qq, t))$.
If~$i$ is integral, then~$\xmid_t = x_i$,~$\xleft_t = x_{i-1}$ and~$\xright_t = x_{i+1}$.
Otherwise,~$\xmid_t$ does not exist, and~$\xleft_t= x_{i-1/2}$ and~$\xright_t = x_{i+1/2}$.
For easier argumentation, if~$\xmid_t$ does not exist, then we say that~$\Pend[\xmid_t] = \sigma_{t \xmid_t}(v)/\sigma_{t \xmid_t} = 0$, and hence,~$\gamma(\xmid_t, t, v) = 0$.


\subsubsection{Vertices outside of the \maxpaths{}}

We now show how shortest paths between two fixed \maxpaths{}~$\Pmax_1$ and~$\Pmax_2$ affect the \bc{} of vertices that are not contained in~$\Pmax_1$ or in~$\Pmax_2$, that is~$v \in V(G) \setminus (V(\Pmax_1) \cup V(\Pmax_2))$.
Note that in the course of the algorithm, we first gather values in~$\Inc[\cdot,\cdot]$ and in the final step compute for each~$s,t \in \Vthree(G)$ the values~$\Inc[s, t] \cdot \sigma_{s t}(v)$ in~$O(m)$ time (\Cref{lem:brandesalg}).
This postprocessing (see \Cref{line:modified-bfs-loop,line:modified-bfs} in \Cref{alg:overview}) can be run in~$O(kn)$ time.
To keep the following proofs simple we assume that these values~$\Inc[s, t] \cdot \sigma_{s t}(v)$ can be computed in constant time for every~$s, t \in \Vthree(G)$ and~$v \in V(G)$.

\begin{lemma}
	\label[lemma]{lem:bc-two-mp-outside}
	Let~$\Pmax_1 \ne \Pmax_2 \in \AllPmax$.
	Then, assuming that the values~$d_G(s, t)$, $\Wleft[v]$ and~$\Wright[v]$ are known for~$s, t \in \Vthree(G)$ and~$v \in \Vtwo(G)$ respectively, and that the values~$\Inc[s, t] \cdot \sigma_{s t}(v)$ can be computed in constant time for every~$s, t \in \Vthree(G)$ and~$v \in V(G)$, one can compute the following for all~$v \in V(G) \setminus (V(\Pmax_1 \cup V(\Pmax_2))$ in~$O(|V(\Pmax_2)|)$ time:
	\begin{equation}
		\label{eq:bc-two-mp-outside}
		\sum_{s \in \Vtwo(\Pmax_1), t \in \Vtwo(\Pmax_2)} \gamma(s, t, v).
	\end{equation}
\end{lemma}
\begin{proof}
	We fix~$\Pmax_1 \ne \Pmax_2 \in \AllPmax$ with~$\Pmax_1 = x_0 \dots x_\qq$
	and~$\Pmax_2 = y_0 \dots y_\qr$.
	We show how to compute~$\sum_{s \in \Vtwo(\Pmax_1)} \gamma(s, t, v)$
	for a fixed~$t \in \Vtwo(\Pmax_2)$ and~$v \in V(G) \setminus (V(\Pmax_1) \cup V(\Pmax_2))$.
	Afterwards, we analyze the running time.

	By definition of~$\xleft_t$,~$\xmid_t$ and~$\xright_t$ we have
	\begin{equation} \label{eq:bc-two-mp-split}
		\sum_{s \in \Vtwo(\Pmax_1)} \gamma(s, t, v) = \gamma(\xmid_t, t, v) + \sum_{s \in \Xleft_t} \gamma(s, t, v) + \sum_{s \in \Xright_t} \gamma(s, t, v).
	\end{equation}

	By definition of \maxpaths{}, every shortest path from~$s \in \Vtwo(\Pmax_1)$ to~$t$ visits either~$y_0$ or~$y_\qr$.
	For~$\psi \in \{x_0, x_\qq\}$ let~$S(t, \psi)$ be a maximal subset of~$\{y_0, y_\qr\}$ such that for each~$\varphi \in S(t, \psi)$ there is a shortest~$st$-path via~$\psi$ and~$\varphi$.
	An example for this notation is given in \cref{fig:v-outside-two-paths}.
	\begin{figure}[t]
		\begin{minipage}{0.5\textwidth}
			\begin{tikzpicture}
				\begin{scope}
					\node (pmax) at (2,1.35) {$\Pmax_1$};
					\node (pmax) at (2,-1.30) {$\Pmax_2$};
					\node[vertex,label=below:{$\varphi$}] (x0) at (0,1.0) {};
					\node[vertex,label=below:{$\bar\varphi$}] (xq) at (4,1.0) {};

					\node[vertex,label=below:{$y_0$}] (y0) at (0,-1.0) {};
					\node[vertex,label=below:{$y_r$}] (yr) at (4,-1.0) {};

					\node[vertex,label=below:{$s$}] (s) at (1.3,-1.0) {};
					\node[vertex,label=below:{$t$}] (t) at (3.1,1.0) {};

					\draw[thick] (x0.center) -- (t.center);
					\draw[thick] (t.center) -- (xq);
					\draw[thick] (y0.center) -- (s);
					\draw[thick] (s) -- (yr.center);

					\node[vertex,fill,black,label=left:{$v$}] at (-0.587,0) {};
					
					\tikzstyle{path} = [color=black!30,line cap=round, line join=round, line width=2.4mm]
					\begin{pgfonlayer}{background}
						\draw[path] (s.center) -- (y0.center) to[out=180,in=180] (x0.center) -- (t.center);
						\draw[path,opacity=1,line width=0.9mm,color=white,out=180,in=180] (y0.center) to (x0.center);
						\draw[dash pattern=on 3pt off 3pt,thick,black!70,out=180,in=180] (y0.center) to (x0.center);
					\end{pgfonlayer}
				\end{scope}
			\end{tikzpicture}
		\end{minipage}
		\hfill
		\begin{minipage}{0.48\textwidth}
			\caption{
				\label{fig:v-outside-two-paths}
				An example for the proof of \cref{lem:bc-two-mp-outside}.
				The endpoints of~$\Pmax_1$ are~$\varphi$ and~$\bar\varphi$.
				In this example we have~$s \in \Xleft_t$, and the set~$S(t, x_0) = \{\varphi\}$.
				Hence, every shortest path from~$s$ to~$t$ visits~$y_0$ and~$\varphi$.
			}
		\end{minipage}
	\end{figure}
	Then, for~$s \in \Xleft_t$, all~$st$-paths visit~$x_0$ and~$\varphi \in S(t, x_0)$.
	Hence, we have that~$\sigma_{s t} = \sum_{\varphi \in S(t, x_0)} \sigma_{x_0 \varphi}$ and~$\sigma_{s t}(v) = \sum_{\varphi \in S(t, x_0)} \sigma_{x_0 \varphi}(v)$.
	Analogously, for~$s \in \Xright_t$ we have that~$\sigma_{s t} = \sum_{\varphi \in S(t, x_\qq)} \sigma_{x_\qq \varphi}$ and~$\sigma_{s t}(v) = \sum_{\varphi \in S(t, x_\qq)} \sigma_{x_\qq \varphi}(v)$.
	Paths from~$t$ to~$\xmid_t$ may visit~$x_0$ and~$\varphi \in S(t, x_0)$ or~$x_\qq$ and~$\varphi \in S(t, x_\qq)$.
	Hence,~$\sigma_{t \xmid_t} = \sum_{\varphi \in S(t x_0)} \sigma_{x_0 \varphi} + \sum_{\varphi \in S(t, x_\qq)} \sigma_{x_\qq \varphi}$.
	The equality holds analogously for~$\sigma_{t \xmid_t}(v)$.
	With this at hand, we can simplify the computation of the first sum of \Cref{eq:bc-two-mp-split}:
	\begin{align*}
		\sum_{s \in \Xleft_t} \gamma(s, t, v)
		& = \sum_{s \in \Xleft_t} \Pend[s] \cdot \Pend[t] \cdot \frac{\sigma_{s t}(v)}{\sigma_{s t}}\\
			& = \Big(\sum_{s \in \Xleft_t} \Pend[s] \Big) \cdot \Pend[t] \cdot
			\frac{\sum_{\varphi \in S(t, x_0)} \sigma_{x_0 \varphi}(v)} {\sum_{\varphi \in S(t, x_0)} \sigma_{x_0 \varphi}} \\
		& = \Wleft[\xleft_t] \cdot \Pend[t] \cdot \frac{\sum_{\varphi \in S(t, x_0)} \sigma_{x_0 \varphi}(v)}{\sum_{\varphi \in S(t, x_0)} \sigma_{x_0 \varphi}}.
			\stepcounter{equation}\tag{\theequation}\label{eq:bc-two-mp-wleft}\\
	\intertext{Analogously,}
		\sum_{s \in \Xright_t} \gamma(s, t, v)
		&= \Wright[\xright_t] \cdot \Pend[t] \cdot
			\frac{\sum_{\varphi \in S(t, x_\qq)} \sigma_{x_\qq \varphi}(v)}
				{\sum_{\varphi \in S(t, x_\qq)} \sigma_{x_\qq \varphi}},
			\stepcounter{equation}\tag{\theequation}\label{eq:bc-two-mp-wright}
	\end{align*}
	and
	\begin{align*}
		\gamma(\xmid_t, t, v)
		&= \Pend[\xmid_t] \cdot \Pend[t] \cdot
			\frac{\sum_{\varphi \in S(t, x_0)} \sigma_{x_0 \varphi}(v)
				+ \sum_{\varphi \in S(t, x_\qq)} \sigma_{x_\qq \varphi}(v)}
			{\sum_{\varphi \in S(t, x_0)} \sigma_{x_0 \varphi}
			+ \sum_{\varphi \in S(t, x_\qq)} \sigma_{x_\qq \varphi}}
			\stepcounter{equation}\tag{\theequation}\label{eq:bc-two-mp-xmid}.
	\end{align*}
	With this we can rewrite \Cref{eq:bc-two-mp-split} to
	\begin{align*}
		&\sum_{s \in \Vtwo(\Pmax_1)} \gamma(s, t, v)\\
		\overset{\eqref{eq:bc-two-mp-wleft},\eqref{eq:bc-two-mp-wright},\eqref{eq:bc-two-mp-xmid}}
		=&~ \frac{\Wleft[\xleft_t] \cdot \Pend[t]}
			{\sum_{\varphi \in S(t, x_0)} \sigma_{x_0 \varphi}}
			\cdot \sum_{\varphi \in S(t, x_0)} \sigma_{x_0 \varphi}(v)\\
		+&~ \frac{\Wright[\xright_t] \cdot \Pend[t]}
			{\sum_{\varphi \in S(t, x_\qq)} \sigma_{x_\qq \varphi}}
			\cdot \sum_{\varphi \in S(t, x_\qq)} \sigma_{x_\qq \varphi}(v)\\
		+&~ \Pend[\xmid_t] \cdot \Pend[t] \cdot
			\frac{\sum_{\varphi \in S(t, x_0)} \sigma_{x_0 \varphi}(v)
				+ \sum_{\varphi \in S(t, x_\qq)} \sigma_{x_\qq \varphi}(v)}
			{\sum_{\varphi \in S(t, x_0)} \sigma_{x_0 \varphi}
			+ \sum_{\varphi \in S(t, x_\qq)} \sigma_{x_\qq \varphi}}.
	\end{align*}
	By joining values~$\sigma_{x_0 \varphi}(v)$ and~$\sigma_{x_\qq \varphi}(v)$ we obtain
	\begin{align*}
		&\sum_{s \in \Vtwo(\Pmax_1)} \gamma(s, t, v)\\
		\overset{~~~~}
		=&~\Big(\tfrac{\Wleft[\xleft_t] \cdot \Pend[t]}
				{\sum_{\varphi \in S(t, x_0)} \sigma_{x_0 \varphi}}
				+ \tfrac{\Pend[\xmid_t] \cdot \Pend[t]}
				{\sum_{\varphi \in S(t, x_0)} \sigma_{x_0 \varphi}
				+ \sum_{\varphi \in S(t, x_\qq)} \sigma_{x_\qq \varphi}}
			\Big) \cdot \sum_{\varphi \in S(t, x_0)} \!\!\!\!\sigma_{x_0 \varphi}(v)
			\stepcounter{equation}\tag{\theequation}\label{eq:bc-two-mp-left}\\
		+&~\Big(\tfrac{\Wright[\xright_t] \cdot \Pend[t]}
				{\sum_{\varphi \in S(t, x_\qq)} \sigma_{x_\qq \varphi}}
				+ \tfrac{\Pend[\xmid_t] \cdot \Pend[t]}
				{\sum_{\varphi \in S(t, x_0)} \sigma_{x_0 \varphi}
				+ \sum_{\varphi \in S(t, x_\qq)} \sigma_{x_\qq \varphi}}
			\Big) \cdot \sum_{\varphi \in S(t, x_\qq)} \!\!\!\!\sigma_{x_\qq \varphi}(v)
			\stepcounter{equation}\tag{\theequation}\label{eq:bc-two-mp-right}\\
		\overset{~~~~}
		\eqqcolon &~X_1 \cdot \sum_{\varphi \in S(t, x_0)} \sigma_{x_0 \varphi}(v)
			+ X_2 \cdot \sum_{\varphi \in S(t, x_\qq)} \sigma_{x_\qq \varphi}(v).
	\end{align*}
	Note that we define~$X_1$ and~$X_2$ to be the terms in the parentheses before the two sums.

	We need to increase the \bc{} of all vertices on shortest paths from~$s$ to~$t$ via~$x_0$ by the value of Term~\eqref{eq:bc-two-mp-left}, and those shortest paths via~$x_\qq$ by the value of Term~\eqref{eq:bc-two-mp-right}.
	By \Cref{lem:brandesalg}, increasing~$\Inc[s, t]$ by some value~$A$ ensures the increment of the \bc{} of~$v$ by~$A \cdot \sigma_{s t}(v)$ for all vertices~$v$ that are on a shortest path between~$s$ and~$t$.
	Hence, increasing~$\Inc[x_0, \varphi]$ for every~$\varphi \in S(t, x_0)$ by~$X_1$
	is equivalent to increasing the \bc{} of~$v$ by the value of Term~\eqref{eq:bc-two-mp-left}.
	Analogously, increasing~$\Inc[x_\qq, \varphi]$ for every~$\varphi \in S(t, x_\qq)$ by~$X_2$
	is equivalent to increasing the \bc{} of~$v$ by the value of Term~\eqref{eq:bc-two-mp-right}.

	We now have incremented~$\Inc[\psi, \varphi]$ for~$\psi \in \{x_0, x_\qq\}$ and~$\varphi \in \{y_0, y_\qr\}$ by certain values, and we have shown that this increment is correct if the shortest~$\psi \varphi$-paths do not visit inner vertices of~$\Pmax_1$ or~$\Pmax_2$.
	We still need to show that
	\begin{inparaenum}[(1)]
		\item increasing~$\Inc[\psi, \varphi]$ does not affect the \bc{} of~$\psi$ or~$\varphi$, and that
		\item we increase~$\Inc[\psi, \varphi]$ only if no shortest~$\psi \varphi$-path visits inner vertices of~$\Pmax_1$ or~$\Pmax_2$.
	\end{inparaenum}

	For the first point, recall that for each~$s, t \in \Vthree(G)$ the \bc{} of~$v \in V(G)$ is increased by~$\Inc[s, t] \cdot \sigma_{s t}(v)$.
	But since~$\sigma_{\psi \varphi}(\psi) = \sigma_{\psi \varphi}(\varphi) = 0$, increments of~$\Inc[\psi, \varphi]$ do not affect the \bc{} of~$\psi$ or~$\varphi$.

	For the second point, suppose that there is a shortest~$\psi \varphi$-path that visits inner vertices of~$\Pmax_2$.
	Let~$\bar{\varphi} \ne \varphi$ be the second endpoint of~$\Pmax_2$.
	Then~$d_G(\psi, \varphi) = d_G(\psi, \bar{\varphi}) + d_G(\bar{\varphi}, \varphi)$, and for all inner vertices~$y_i$ of~$\Pmax_2$, that is, for all~$y_i$ with~$1 \le i < r$, it holds that
	\[ d_G(\psi, \varphi) + d_G(\varphi, y_i) \ =\  d_G(\psi, \bar{\varphi}) + d_G(\bar{\varphi}, \varphi) + d_G(\varphi, y_i) \ >\  d_G(\psi, \bar{\varphi}) + d_G(\bar{\varphi}, y_i). \]
	Hence, there are no shortest~$y_i \psi$-paths that visit~$\varphi$, and consequently~$\Inc[\psi, \varphi]$ will not be incremen\-ted.
	The same argument holds if there is a shortest~$\psi \varphi$-path that visits inner vertices of~$\Pmax_1$.

	Finally, we analyze the running time.
	The values~$\Wleft[\cdot]$,~$\Wright[\cdot]$ and~$\Pend[\cdot]$ as well as the distances and number of shortest paths between all pairs of vertices of degree at least three are assumed to be known.
	With this,~$S(t, x_0)$ and~$S(t, x_\qq)$ can be computed in constant time.
	Hence, the values~$X_1$ and~$X_2$ can be computed in constant time for a fixed~$t \in \Vtwo(\Pmax_2)$.
	Thus, the running time to compute the increments of~$\Inc[\cdot, \cdot]$ is upper-bounded by~$O(|V(\Pmax_2)|)$.
\end{proof}

\subsubsection{Vertices inside the \maxpaths{}}
We now consider how shortest paths between pairs of two \maxpaths{}~$\Pmax_1 \ne \Pmax_2$ affect the \bc{} of their vertices.

When iterating through all pairs~$\Pmax_1 \ne \Pmax_2 \in \AllPmax$, one will encounter the pair~$(\Pmax_1, \Pmax_2)$ and its reverse~$(\Pmax_2, \Pmax_1)$.
Since our graph is undirected, instead of looking at the \bc{} of the vertices in both \maxpaths{}, it suffices to consider only the vertices inside the second \maxpath{} of the pair.
This is shown in the following lemma.
\begin{lemma}
	\label[lemma]{lem:bc-two-mp-inside-two}
	Computing for each~$\Pmax_1 \ne \Pmax_2 \in \AllPmax$ and
	for each vertex~${v \in V(\Pmax_1) \cup V(\Pmax_2)}$
	\begin{equation*}
		\sum_{s \in \Vtwo(\Pmax_1),\, t \in \Vtwo(\Pmax_2)} \gamma(s, t, v)
	\end{equation*}
	is equivalent to computing for every~$\Pmax_1 \ne \Pmax_2 \in \AllPmax$ and for each~$v \in V(\Pmax_2)$
	\begin{equation}
		\label{eq:Xv}
		X_v = 
		\begin{cases}
			\phantom{2 \cdot}
			\displaystyle \sum_{s \in \Vtwo(\Pmax_1), t \in \Vtwo(\Pmax_2)}
				\gamma(s, t, v), & \text{if } v \in V(\Pmax_1) \cap V(\Pmax_2);\\
			2 \cdot
			\displaystyle \sum_{s \in \Vtwo(\Pmax_1), t \in \Vtwo(\Pmax_2)}
				\gamma(s, t, v), & \text{otherwise.}
		\end{cases}
	\end{equation}
\end{lemma}
Since the proof of \cref{lem:bc-two-mp-inside-two} is rather tedious and not special in terms of employed methods, we defer it to \Cref{sec:proofs-two-maxpaths}.

With this at hand we can show how to compute~$X_v$ for each~$v \in V(\Pmax_2)$, for a
pair~$\Pmax_1 \ne \Pmax_2 \in \AllPmax$ of \maxpaths{}.
To this end, we show the following lemma.
\begin{lemma}
	\label[lemma]{lem:bc-two-mp-inside}
	Let~$\Pmax_1 \ne \Pmax_2 \in \AllPmax$.
	Then, given that the values~$d_G(s, t)$, $\sigma_{s t}$, $\Wleft[v]$ and~$\Wright[v]$ are known for~$s, t \in \Vthree(G)$ and~$v \in \Vtwo(G)$, respectively, one can compute for all~$v \in V(\Pmax_2)$ in~$O(|V(\Pmax_2)|)$ time:
	\begin{equation}
		\label{eq:bc-two-mp-inside}
		\sum_{s \in \Vtwo(\Pmax_1),\, t \in \Vtwo(\Pmax_2)} \gamma(s, t, v).
	\end{equation}
\end{lemma}
Again, due to being tedious, we defer the proof of \cref{lem:bc-two-mp-inside} to \Cref{sec:proofs-two-maxpaths-inside}.
The high-level approach has two steps:
First, we show how to compute the value~$\sum_{s \in \Vtwo(\Pmax_1)} \gamma(s, t, v)$ for a fixed~$t \in \Vtwo(\Pmax_2)$ and~$v \in V(\Pmax_2)$ in constant time; here we use that the values listed in the lemma are known.
Second, we use a dynamic program to compute for all~$v \in V(\Pmax_2)$ the value of Sum~\eqref{eq:bc-two-mp-inside} in~$O(|V(\Pmax_2)|)$ time, using the fact that the difference between the sums of two adjacent~$v, v' \in V(\Pmax_2)$ can be computed in constant time.

We are now ready to combine \Cref{lem:bc-two-mp-outside,lem:bc-two-mp-inside-two,lem:bc-two-mp-inside} to prove \Cref{prop:mp-pairs}.
As mentioned above, to keep the proposition simple, we assume that the values~$\Inc[s, t] \cdot \sigma_{s t}(v)$ can be computed in constant time for every~$s, t \in \Vthree(G)$ and~$v \in V(G)$.
In fact, these values are computed in the last step of the algorithm (see \Cref{line:modified-bfs-loop,line:modified-bfs} in \Cref{alg:overview} and \Cref{lem:brandesalg}).
\begingroup
\def\theproposition{\ref{prop:mp-pairs}}
\begin{proposition}[Restated]
	\propmppairs
\end{proposition}
\addtocounter{proposition}{-1}
\endgroup
\begin{proof}
	Let~$\Pmax_1 \ne \Pmax_2 \in \AllPmax$.
	Then, for each~$v \in V(G) = (V(G) \setminus (V(\Pmax_1) \cup V(\Pmax_2)))
	\cup (V(\Pmax_1) \cup V(\Pmax_2))$, we need to compute
	\begin{equation}
		\label{eq:bc-two-mp-sum}
		\sum_{s \in \Vtwo(\Pmax_1), t \in \Vtwo(\Pmax_2)} \gamma(s, t, v).
	\end{equation}
	We first compute in~$O(kn)$ time the values~$d_G(s, t)$ and~$\sigma_{s t}$ for every~$s, t \in \Vthree(G)$, as well as the values~$\Wleft[v]$ and~$\Wright[v]$ for every~$v \in \Vtwo(G)$, see \Crefrange{line:init-tables}{line:compute-wright} in \Cref{alg:overview}.
	By \Cref{lem:bc-two-mp-outside} we then can compute Sum~\eqref{eq:bc-two-mp-sum} in~$O(|V(\Pmax_2)|)$ time for~$v \in V(G) \setminus (V(\Pmax_1) \cup V(\Pmax_2))$.
	Given the values~$\rho_i$ of \Cref{lem:bc-two-mp-inside} we can compute the values~$X_v$ defined in \Cref{eq:Xv} for~$v = y_i \in V(\Pmax_2)$ as follows:
	\begin{equation*}
		X_v = X_{y_i} =
			\begin{cases}
				\rho_i, & \text{if } v \in V(\Pmax_1) \cap V(\Pmax_2);\\
				2 \rho_i, & \text{otherwise.}
			\end{cases}
	\end{equation*}
	This can be done in constant time for a single~$v \in V(\Pmax_2)$; thus it can be done in~$O(|V(\Pmax_2)|)$ time overall.
	Hence, by \Cref{lem:bc-two-mp-inside-two}, we can compute Sum~\eqref{eq:bc-two-mp-sum} for~$V(\Pmax_1) \cup V(\Pmax_2)$ in~$O(|V(\Pmax_2)|)$ time.

	Sum~\eqref{eq:bc-two-mp-sum} must be computed for every
	pair~$\Pmax_1 \ne \Pmax_2 \in \AllPmax$.
	Thus, overall, we require
	\begin{align*}
		&~O\Big(\sum_{\Pmax_1 \ne \Pmax_2 \in \AllPmax} |V(\Pmax_2)|\Big)\\
		=&~O\Big(\sum_{\Pmax_1 \in \AllPmax} 
			\sum_{\substack{\Pmax_2 \in \AllPmax\\ \Pmax_1 \ne \Pmax_2}}
		\big(|\Vtwo(\Pmax_2)| + |\Vthree(\Pmax_2)|\big)\Big)\\
		\stepcounter{equation}\tag{\theequation}\label{eq:bc-mp-inside-time}
		=&~ O\Big(\sum_{\Pmax_1 \in \AllPmax} n \Big) = O(kn)
	\end{align*}
	time, since there are at most~$O(k)$ \maxpaths{} and at most~$n$ vertices in all \maxpaths{} combined.
\end{proof}

\subsection{Paths with endpoints in the same \maxpath{}}
\label{sec:one-maxpath}
We now look at shortest paths starting and ending in a \maxpath{}~$\Pmax = x_0 \dots x_\qq$ and show how to efficiently compute how these paths affect the \bc{} of all vertices in the graph.
The goal is to prove the following:

\begin{proposition}
	\label[proposition]{prop:mp-single}
	In~$O(kn)$ time one can compute the following for all~$v \in V(G)$:
	\begin{equation*}
		\sum_{\substack{s, t \in \Vtwo(\Pmax)\\	\Pmax \in \AllPmax}}\gamma(s, t, v).
	\end{equation*}
\end{proposition}

We start off by noting the following:
\begin{observation}
	Let~$v \in V(G)$ and let~$\Pmax = x_0\ldots x_\qq$ be a \maxpath{}.
	Then
	\[\sum_{\substack{s, t \in \Vtwo(\Pmax)}}\gamma(s, t, v) = \sum_{\substack{i, j \in [1,\qq-1]}}\gamma(x_i, x_j, v) = 2\cdot\sum_{i=1}^{\qq-1}\sum_{j = i + 1} ^{\qq-1}\gamma(x_i, x_j, v).\]
\end{observation}

For the sake of readability we set~$[x_p,x_r] := \{x_p, x_{p+1}, \ldots, x_r\}$, $p<r$.
We will distinguish between two different cases that we then treat separately:
Either~$v \in [x_i,x_j]$ or~$v \in V(G) \setminus [x_i,x_j]$.
We will show that both cases can be solved in overall~$O(|V(\Pmax)|)$ time for~$\Pmax$.
Doing this for all \maxpaths{} results in a time of~$O(\sum_{\Pmax \in \AllPmax} |\Vtwo(\Pmax)|) = O(n)$.
In the calculations we will distinguish between the two main cases---all shortest~$x_ix_j$-paths are fully contained in~$\Pmax$, or all shortest~$x_ix_j$-paths leave~$\Pmax$---and the corner case that there are some shortest paths inside~$\Pmax$ and some that partially leave it.

We will now compute the value for all paths that only consist of vertices in~$\Pmax$, that is, we will compute for each~$x_k$ with~$i<k<j$ the term
\[
	2\cdot\sum_{i=1}^{\qq-1}\sum_{j = i+1}^{\qq-1}\gamma(x_i, x_j, x_k)
\]
with a dynamic program in~$O(|V(\Pmax)|)$ time.
Since~$i<k<j$, by \cref{obs:gamma}, this can be simplified to
\[
	2\cdot\sum_{\substack{i\in[1,\qq-1]\\i<k}}\sum_{\substack{j \in [i+1,\qq-1]\\k<j}}\gamma(x_i, x_j, x_k)
	= 2\cdot\sum_{i\in[1,k-1]}\sum_{j \in [k+1,\qq-1]}\gamma(x_i, x_j, x_k).
\]
\begin{lemma}
	\label[lemma]{lem:v-in-one-maxpath}
	Let~$\Pmax = x_0 \ldots x_\qq$ be a \maxpath{}.
	In~$O(|V(\Pmax)|)$ time, one can compute the following for all~$x_k$ with~$0 \le k \le \qq$:
	\[\alpha_{x_k} \coloneqq 2\cdot\sum_{i\in[1,k-1]}\sum_{j \in [k+1,\qq-1]}\gamma(x_i, x_j, x_k).\]
\end{lemma}
The main idea of the dynamic program is the following:
Given the value of~$\alpha_{x_k}$, one can compute its difference to~$\alpha_{x_{k+1}}$ in constant time, once~$\Wleft, \Wright$ are precomputed (see \Crefrange{line:wleftright}{line:compute-wright} in \cref{alg:overview}).
These tables can be computed in~$O(|V(\Pmax)|)$ time as well.
The proof of \cref{lem:v-in-one-maxpath} is deferred to \cref{sec:proofs-one-maxpath}.

Now we need to show how to compute the value for all paths that (partially) leave~$\Pmax$.
See \cref{fig:v-partially-in-path} for an example of such a path.
\begin{figure}[t]
	\centering
	\begin{tikzpicture}
		\begin{scope}
			\node (pmax) at (2,-1.00) {$\Pmax$};
			\node[vertex,label=below:{$x_0$}] (x0) at (0,-.75) {};
			\node[vertex,label=below:{$x_q$}] (xq) at (4,-.75) {};

			\node[vertex,label=below:{$s$}] (s) at (1.3,-.75) {};
			\node[vertex,label=below:{$t$}] (t) at (3.1,-.75) {};

			\draw[thick] (t.center) -- (s);
			\draw[thick] (x0) -- (s);
			\draw[thick] (t) -- (xq);

			\node[vertex,fill,black,label=above:{$v$}] at (2, .075) {};
			
			\tikzstyle{path} = [color=black!30,line cap=round, line join=round, line width=2.4mm]
			\begin{pgfonlayer}{background}
				\draw[path] (s.center) -- (x0.center) to[out=135,in=45] (xq.center) -- (t.center);
				\draw[path,opacity=1,line width=0.9mm,color=white,out=135,in=45] (x0.center) to (xq.center);
				\draw[dash pattern=on 3pt off 3pt,thick,black!70,out=135,in=45] (x0.center) to (xq.center);
			\end{pgfonlayer}
		\end{scope}
	\end{tikzpicture}
	\caption{
		\label{fig:v-partially-in-path}
		A \maxpath{} that affects the betweenness centralities of vertces outside of~$\Pmax$, such as~$v$.
		Clearly, if there is a shortest~$st$-path visiting~$v$ (thick edge), then there exists a shortest~$x_0 x_{\qq}$-path visiting~$v$ (dashed edge).
		On an intuitive level, we store the information of the vertices inside of~$\Pmax$ in the table entry~$\Inc[x_0, x_{\qq}]$.
	}
\end{figure}

\begin{lemma}
	\label[lemma]{lem:v-not-within-maxpath}
	Let~$\Pmax = x_0 x_1 \dots x_\qq$ be a \maxpath{}.
	Then, assuming that~$\Inc[s, t] \cdot \sigma_{s t}(v)$ can be computed in constant time for some~$s, t \in \Vthree(G)$, one can compute in~$O(|V(\Pmax)|)$ time the following for all~$v \in V(G) \setminus [x_i, x_j]$:
 	\[\beta_v := \sum_{i\in[1,\qq-1]}\sum_{j \in [i+1,\qq-1]}\gamma(x_i, x_j, v).\]
\end{lemma}
Note that in the postprocessing (see \Cref{line:modified-bfs-loop,line:modified-bfs} in \Cref{alg:overview}) the \bc{} value of each vertex~$v$ is increased by~$\Inc[s, t] \cdot \sigma_{s t}(v)$ for each pair of vertices~$s, t \in \Vthree(G)$.
The proof is split into two cases:
Either~$v \in V(G) \setminus V(\Pmax)$, or~$v \in V(\Pmax) \setminus [x_i, x_j]$ (the case that~$v \in [x_i, x_j]$ is covered by \cref{lem:v-in-one-maxpath}).
The first case makes use of the postprocessing step (see \Crefrange{line:modified-bfs-loop}{line:modified-bfs} in \cref{alg:overview}) which was used in an analogous way in the proof of \cref{lem:bc-two-mp-outside}, while the second case uses a dynamic programming approach similar to the one used in the proof of \cref{lem:v-in-one-maxpath}.
The proof details can be found in \Cref{sec:proofs-one-maxpath-not-within}.


\subsection{Postprocessing and algorithm summary}
\label{sec:postproc-summary}
We are now ready to combine all parts and prove our main theorem.
\begin{theorem}\label{thm:main}
	\bctask{} can be solved in~$O(k n)$ time and space, where~$k$ is the feedback edge number of the input graph.
\end{theorem}
\begin{proof}
	As shown in \cref{prop:pendingcycles}, if the input graph~$G$ is a cycle, then we are done.

	We show that \cref{alg:overview} computes the value
	\begin{align*}
		C_B(v) = \sum_{s,t \in V(G)} \Pend[s] \cdot \Pend[t] \cdot \frac{\sigma_{st}(v)}{\sigma_{st}} = \sum_{s,t \in V(G)} \gamma(s,t,v)
	\end{align*}
	for all~$v \in V(G)$ in~$O(kn)$ time and space.
	We use \Cref{obs:split-bc-sum} to split the sum as follows.
	\begin{align*}
		\sum_{s,t \in V(G)}
		\gamma(s, t, v)
		&= \sum_{s \in \Vthree(G),\, t \in V(G)}
		\gamma(s, t, v)
		+ \sum_{s \in \Vtwo(G),\, t \in \Vthree(G)}
		\gamma(t, s, v)\\&
		+ \sum_{\substack{
			s \in \Vtwo(\Pmax_1),\, t \in \Vtwo(\Pmax_2)\\
			\Pmax_1 \ne \Pmax_2 \in \AllPmax}}
		\gamma(s, t, v)
		+ \sum_{\substack{
			s, t \in \Vtwo(\Pmax)\\
			\Pmax \in \AllPmax}}
		\gamma(s, t, v).
	\end{align*}

	By \Cref{prop:mp-pairs,prop:mp-single}, we can compute the third and fourth summand in~$O(kn)$ time provided that~$\Inc[s, t] \cdot \sigma_{s t}(v)$ is computed for every~$s, t \in \Vthree(G)$ and every~$v \in V(G)$ in a postprocessing step (see \Crefrange{line:maxpaths-loop}{line:inc-update-one-path}). 
	We incorporate this postprocessing into the computation of the first two summands in the equation, that is, we next show that for all~$v \in V(G)$ the following value can be computed in~$O(kn)$ time:
	\[\sum\limits_{\substack{s \in \Vthree(G)\\ t \in V(G)}}\gamma(s, t, v) + \sum\limits_{\substack{s \in \Vtwo(G)\\ t \in \Vthree(G)}}\gamma(s, t, v) + \sum\limits_{\substack{s \in \Vthree(G)\\ t \in \Vthree(G)}}\Inc[s,t]\cdot\sigma_{st}(v).\]
	To this end, observe that the above is equal to
	\begin{align*}
		& \smashoperator[r]{\sum\limits_{\substack{s \in \Vthree(G)\\ t \in V(G)}}}\Pend[s] \Pend[t]\tfrac{\sigma_{st}(v)}{\sigma_{st}} + \smashoperator[r]{\sum\limits_{\substack{s \in \Vthree(G)\\ t \in \Vtwo(G)}}}\Pend[s] \Pend[t]\tfrac{\sigma_{st}(v)}{\sigma_{st}}
		+ \smashoperator[r]{\sum\limits_{\substack{s \in \Vthree(G)\\ t \in \Vthree(G)}}}\Inc[s,t]\sigma_{st}(v)\\
		= & \sum\limits_{s \in \Vthree(G)} \!\Big ( (2 \cdot \smashoperator[r]{\sum\limits_{t \in \Vtwo(G)}} \Pend[s] \Pend[t]\tfrac{\sigma_{st}(v)}{\sigma_{st}}) + \smashoperator[r]{\sum\limits_{t\in \Vthree}} \sigma_{st}(v)  (\tfrac{\Pend[s] \Pend[t]}{\sigma_{st}} + \Inc[s,t]) \Big ).
	\end{align*}
	Note that we initialize~$\Inc[s,t]$ in \Cref{line:init-inc-table,line:init-inc-table-vtwo} in \Cref{alg:overview} with~$2\cdot \Pend[s] \Pend[t] / \sigma_{st}$ and~$\Pend[s] \cdot\allowbreak \Pend[t]/ \sigma_{st}$ respectively.
	Thus we can use the algorithm described in \Cref{lem:brandesalg} for each vertex~$s\in \Vthree(G)$ with~$f(s,t) = \Inc[s,t]$.

	Since~$\Pend[s],\, \Pend[t],\, \sigma_{st}$ and~$\Inc[s,t]$ can all be looked up in constant time, the algorithm only takes~$O(n+m)$ time for each vertex~$s$ (see \Cref{line:modified-bfs-loop,line:modified-bfs}).
	By \Cref{lem:fen} there are~$O(\min \{k,n\})$ vertices of degree at least three.
	Thus, altogether, the algorithm needs~$O(\min\{n,k\} \cdot m) = O(\min\{n,k\} \cdot (n + k)) = O(kn)$ time.
	The precomputations in \Crefrange{line:init-tables}{line:init-inc-table} require~$\Theta(kn)$ space.
	As the running time is an upper bound on the space complexity, \cref{alg:overview} requires~$\Theta(kn)$ space overall.
\end{proof}

\section{Conclusion}
\label{sec:conclusion}
Lifting the processing of degree-one vertices due to \citet{BGPL12} to a technically much more involved processing of degree-two vertices, we derived a new algorithm for \textsc{Betweenness Centrality} running in $O(kn)$ worst-case time ($k$~is the feedback edge number of the input graph). 
Our work focuses on algorithm theory and contributes to the field of adaptive algorithm design \cite{CastroW92} as well as to the recent ``FPT in P'' field \cite{GMN17}.
It would be of high interest to identify structural parameterizations ``beyond'' the feedback edge number that might help to get more results in the spirit of our work.
In particular, extending our algorithmic approach and mathematical analysis with respect to the treatment of twin vertices~\cite{PEZDB14,SKSC17} might help to get a running time bound involving the vertex cover number of the input graph.
As for practical relevance, we firmly believe that a running time of~$O(kn)$ as we proved can yield improved performance for some real-world networks.
What remains unclear, however, is whether the constants hidden in the~$O$-notation or the non-linear space requirements of our approach can be avoided.

\bibliographystyle{plainnat}
\bibliography{../bib}

\clearpage
\appendix

\section{Notation for proofs in appendix}
For the following proofs we will introduce a lot of auxiliary notation.
We provide \cref{tab:allnotation} as a reference to the definitions of the notations.
\begin{table}[h!]
	\centering
	\caption{A reference to the notation used in the Appendix.\newline
		We assume~$\Pmax = \Pmax_1 = x_0 \dots x_\qq$ and~$\Pmax_2 = y_0 \dots y_\qr$.
	}
	\label{tab:allnotation}
	\begin{tabular}{l r p{9.01cm}}
		\toprule
		Symbol & & Definition\\
		\midrule
		$\gamma_{st}(v)$	&$=$& $\Pend[s] \cdot \Pend[t] \cdot \sigma_{st}(v) / \sigma_{st}$;\\
		$\Vtwo(G)$		&   & the set of vertices of degree two in~$G$;\\
		$\Vthree(G)$		&   & the set of vertices of degree at least three in~$G$;\\
		$\Inc[\cdot, \cdot]$	&   & a table of size~$|\Vthree(G)| \times |V(G)|$ in which intermediary \bc{} values are stored;\\
		$\AllPmax$		&   & the set of all \maxpaths{};\\
		$\xleft_t$		&   & the rightmost vertex in~$\Pmax_1$ such that all shortest paths from~$t \in V(G - \Pmax_1)$ to~$\xleft_t$ visit~$x_0$;\\
		$\xright_t$		&   & the leftmost vertex in~$\Pmax_1$ such that all shortest paths from~$t \in V(G - \Pmax_1)$ to~$\xright_t$ visit~$x_\qq$;\\
		$\xmid_t$		&   & the vertex in~$\Pmax_1$ such that there are shortest paths from~$t \in V(G - \Pmax_1)$ to~$\xleft_t$ via~$x_0$ and~$x_\qq$ respectively;\\
		$\Xleft_t$		&$=$& $\{x_1, x_2, \dots, \xleft_t\}$;\\
		$\Xright_t$		&$=$& $\{\xright_t, \dots x_{\qq - 2}, x_{\qq - 1}\}$;\\
		$\Wleft[x_k]$		&$=$& $\sum_{i=0}^k \Pend[x_i]$, where~$x_i \in \Pmax_1$;\\
		$\Wright[x_k]$		&$=$& $\sum_{i=k}^{\qq-1} \Pend[x_i]$, where~$x_i \in \Pmax_1$;\\
		$S(t, \psi)$		&   & for~$\psi \in \{x_0, x_{\qq} \} = \Vthree(|\Pmax_1|)$, the maximal subset of~$\{y_0, y_{\qr}\} = \Vthree(|\Pmax_2|)$ such that for each~$\varphi \in S(t, \psi)$ there is a shortest~$st$-path via~$\psi$ and~$\varphi$;\\
		$X_v$			&   & see \Cref{eq:Xv};\\
		$\lambda(y_k, y_i)$	&$=$& $\sum_{s \in \Vtwo(\Pmax_1)} \gamma(s, y_k, y_i)$,\\
		&& for $0 \le i \le r$, $1 \le k < r$, and~$s \in \Vtwo(\Pmax_1)$;\\
		$\eta(y_k, \varphi, \psi)$	&   & is $1$ if there is a shortest path from~$y_k$ to~$\psi \in \{x_0, x_{\qq}\}$ to~$\varphi \in \{y_0, y_{\qr}\}$, $0$~otherwise;\\
		$\omega_i$		&   & for~$0 < k, i < \qr$, $y_r$ if~$k < i$, $y_0$ if $k > i$;\\
		$\kappa(y_k, \omega_i)$	&   & see \Cref{eq:kappa};\\
		$\rho_i$		&$=$& $\displaystyle\sum_{k=1}^{i-1} \kappa(y_k, y_\qr) + \sum_{k = i+1}^{\qr-1} \kappa(y_k, y_0)$, for~$0 < i < \qr$;\\
		$[x_i, x_j]$		&$=$& $\{x_i, x_{i+1}, \dots x_j\}$ for~$0 \le i < j \le q$;\\
		$i^+_{\Mid}$		&$=$& $i + (d_G(x_0,x_\qq) + \qq)/{2}$, where~$0 < i < \qq$;\\
		$j^-_{\Mid}$		&$=$& $j - (d_G(x_0,x_\qq) + \qq)/{2}$; where~$0 < j < \qq$;\\
		$\alpha_k$		&$=$& $2 \cdot \sum_{i \in [1, k-1]} \sum_{j \in [k+1, \qq-1]} \gamma(x_i, x_j, x_k)$, where~$0 \le k \le \qq$;\\
		$\beta_v$		&$=$& $\sum_{i \in [1, \qq-1]} \sum_{j \in [i+1, \qq-1]} \gamma(x_i, x_j, v)$,\\
		&& where~$v \in V(G) \setminus [x_i, x_j]$.\\
		\bottomrule
	\end{tabular}
\end{table}

\section{Proofs of Lemmata 7 and 8}
\subsection{Proof of Lemma 7}
\label{sec:proofs-two-maxpaths}

{
\def\thelemma{\ref{lem:bc-two-mp-inside-two}}
\begin{lemma}[Restated] 
	Computing for every pair~$\Pmax_1 \ne \Pmax_2 \in \AllPmax$ and
	for each vertex~${v \in V(\Pmax_1) \cup V(\Pmax_2)}$
	\begin{equation}
		\label{eq:A:bc-two-mp-inside-two}
		\sum_{s \in \Vtwo(\Pmax_1),\, t \in \Vtwo(\Pmax_2)} \gamma(s, t, v)
	\end{equation}
	is equivalent to computing for every~$\Pmax_1 \ne \Pmax_2 \in \AllPmax$ and for each~$v \in V(\Pmax_2)$
	\begin{equation}
		\label{eq:A:Xv}
		X_v = 
		\begin{cases}
			\phantom{2 \cdot}
			\displaystyle \sum_{s \in \Vtwo(\Pmax_1), t \in \Vtwo(\Pmax_2)}
				\gamma(s, t, v), & \text{if } v \in V(\Pmax_1) \cap V(\Pmax_2);\\
			2 \cdot
			\displaystyle \sum_{s \in \Vtwo(\Pmax_1), t \in \Vtwo(\Pmax_2)}
				\gamma(s, t, v), & \text{otherwise.}
		\end{cases}
	\end{equation}
\end{lemma}
\addtocounter{lemma}{-1}
}
\begin{proof}
	We will first assume that~$V(\Pmax_1) \cap V(\Pmax_2) = \emptyset$ for every~$\Pmax_1 \ne \Pmax_2 \in \AllPmax$, and will discuss the special case~$V(\Pmax_1) \cap V(\Pmax_2) \ne \emptyset$ afterwards.

	For every fixed~$\{\Pmax_1, \Pmax_2\} \in \binom{\AllPmax}{2}$ and for every~$v \in V(\Pmax_2)$, the \bc{} of~$v$ is increased by
	\begin{align*}
		& \sum_{s \in \Vtwo(\Pmax_1), t \in \Vtwo(\Pmax_2)} \gamma(s, t, v)
		+ \sum_{s \in \Vtwo(\Pmax_2), t \in \Vtwo(\Pmax_1)} \gamma(s, t, v),
	\end{align*}
	and by \Cref{obs:gamma} this is equal to
	\begin{align}
		\label{eq:A:bc-two-mp-inside-factor}
		2 \cdot \sum_{s \in \Vtwo(\Pmax_1), t \in \Vtwo(\Pmax_2)} \gamma(s, t, v)
	\end{align}
	Analogously, for every~$w \in V(\Pmax_1)$, the \bc{} of~$v$ is increased by
	\begin{align*}
		2 \cdot	\sum_{s \in \Vtwo(\Pmax_1), t \in \Vtwo(\Pmax_2)} \gamma(s, t, w).
	\end{align*}
	Thus, computing Sum~\eqref{eq:A:bc-two-mp-inside-factor} for~$v \in V(\Pmax_2)$ for every
	pair~$\Pmax_1 \ne \Pmax_2 \in \AllPmax$ is equivalent to computing
	Sum~\eqref{eq:A:bc-two-mp-inside-two} for~$v \in V(\Pmax_1) \cup V(\Pmax_2)$ for every
	pair~$\Pmax_1 \ne \Pmax_2 \in \AllPmax$,
	since when iterating over pairs of \maxpaths{} we will encounter both the
	pairs~$(\Pmax_1, \Pmax_2)$ and~$(\Pmax_2, \Pmax_1)$.

	Consider now the special case that there exists a vertex~$v \in V(\Pmax_1) \cap V(\Pmax_2)$.
	Note that this vertex can only be endpoints of~$\Pmax_1$ and~$\Pmax_2$, and it is covered once when performing the computations for~$(\Pmax_1, \Pmax_2)$, and once when performing the computations for~$(\Pmax_2, \Pmax_1)$.
	Hence, we are doing computations twice.
	We compensate for this by increasing the \bc{} of~$v$ only by
	\begin{align*}
		\sum_{s \in \Vtwo(\Pmax_1), t \in \Vtwo(\Pmax_2)}
		\gamma(s, t, v)
	\end{align*}
	for all~$\Pmax_1 \ne \Pmax_2$, for vertices~$v \in V(\Pmax_1) \cap V(\Pmax_2)$.
\end{proof}

\subsection{Proof of Lemma 8}
\label{sec:proofs-two-maxpaths-inside}
{
\def\thelemma{\ref{lem:bc-two-mp-inside}}
\begin{lemma}[Restated] 
	Let~$\Pmax_1 \ne \Pmax_2 \in \AllPmax$.
	Then, given that the values~$d_G(s, t)$, $\sigma_{s t}$, $\Wleft[v]$ and~$\Wright[v]$ are known for~$s, t \in \Vthree(G)$ and~$v \in \Vtwo(G)$, respectively, one can compute for all~$v \in V(\Pmax_2)$ in~$O(|V(\Pmax_2)|)$ time:
	\begin{equation}
		\label{eq:A:bc-two-mp-inside}
		\sum_{s \in \Vtwo(\Pmax_1),\, t \in \Vtwo(\Pmax_2)} \gamma(s, t, v).
	\end{equation}
\end{lemma}
\addtocounter{lemma}{-1}
}
\begin{proof}
	We first show how to compute~$\sum_{s \in \Vtwo(\Pmax_1)} \gamma(s, t, v)$ for fixed~$t \in \Vtwo(\Pmax_2)$ and~$v \in V(\Pmax_2)$ in constant time when the values listed above are known.
	Then we present a dynamic program that computes for all~$v \in V(\Pmax_2)$ the value of Sum~\eqref{eq:A:bc-two-mp-inside} in~$O(|V(\Pmax_2)|)$ time.

	Let~$\Pmax_1 = x_0 \dots x_\qq$ and let~$\Pmax_2 = y_0 \dots y_\qr$.
	For~$v = y_i,\,0 \le i \le \qr$, we compute
	\begin{align*}
		\sum_{s \in \Vtwo(\Pmax_1),\, t \in \Vtwo(\Pmax_2)} \gamma(s, t, y_i)
		&= \sum_{s \in \Vtwo(\Pmax_1)}
			\sum_{k=1}^{\qr-1} \gamma(s, y_k, y_i)\\
		&= \sum_{k=1}^{\qr-1}
			\sum_{s \in \Vtwo(\Pmax_1)} \gamma(s, y_k, y_i).
		\stepcounter{equation}\tag{\theequation}\label{eq:bc-mp-inside-base}
	\end{align*}
	For easier reading, we define for~$0 \le i \le \qr$ and for~$1 \le k < \qr$
	\begin{align*}
		\lambda(y_k, y_i) = \sum_{s \in \Vtwo(\Pmax_1)} \gamma(s, y_k, y_i).
	\end{align*}
	Recall that all shortest paths from~$y_k$ to~$s \in \Xleft_{y_k}$ visit~$x_0$ and all shortest paths from~$y_k$ to~$s \in \Xright_{y_k}$ visit~$x_\qq$.
	Recall also that for each~$y_k$ there may exist a unique vertex~$\xmid_{y_k}$ to which there are shortest paths via~$x_0$ and via~$x_\qq$.

	With this at hand, we have
	\begin{align*}
		\lambda(y_k, y_i)
	 	&= \gamma(\xmid_{y_k}, y_k, y_i) 
			+ \sum_{s \in \Xleft_{y_k}} \gamma(s, y_k, y_i)
			+ \sum_{s \in \Xright_{y_k}} \gamma(s, y_k, y_i)\\
		&= \Pend[\xmid_{y_k}] \cdot \Pend[y_k] \cdot
			\frac{\sigma_{y_k \xmid_{y_k}}(y_i)}{\sigma_{y_k \xmid_{y_k}}}
		+ \sum_{s \in \Xleft_{y_k}}
			\Pend[s] \cdot \Pend[y_k] \cdot \frac{\sigma_{s y_k}(y_i)}{\sigma_{s y_k}}\\
		&\qquad\qquad\qquad\qquad\qquad + \sum_{s \in \Xright_{y_k}}
			\Pend[s] \cdot \Pend[y_k] \cdot\frac{\sigma_{s y_k}(y_i)}{\sigma_{s y_k}}
		\stepcounter{equation}\tag{\theequation}\label{eq:bc-mp-inside-lambda}\\
	\end{align*}

	Next, we rewrite~$\lambda$ in such a way that we can compute it in constant time.
	To this end, we need to make the values~$\sigma$ independent of~$s$ and~$y_i$.
	To this end, note that if~$k < i$, then~$y_i$ is visited only by shortest paths from~$y_k$ to~$s \in \Vtwo(\Pmax_1)$ that also visit~$y_\qr$.
	If~$k > i$, then~$y_i$ is only visited by paths that also visit~$y_0$.
	Hence, we need to know whether there are shortest paths from~$y_k$ to some endpoint of~$\Pmax_1$ via either~$y_0$ or~$y_\qr$.
	For this we define~$\eta(y_k, \varphi, \psi)$, which, informally speaking, tells us whether there is a shortest path from~$y_k$ to~$\psi \in \{x_0, x_\qq\}$ via~$\varphi \in \{y_0, y_\qr\}$.
	Formally,
	\begin{align*}
		\eta(y_k, \varphi, \psi) = \begin{cases}
			1, & \text{if } d_G(y_k, \varphi) + d_G(\varphi, \psi) = d_G(y_k, \psi);\\
			0, & \text{otherwise.}
		\end{cases}
	\end{align*}
	Since~$d_G(s, t)$ is given for all~$s, t \in \Vthree(G)$, the values~$\eta$ can be computed in constant time.

	We now show how to compute~$\sigma_{s y_k}(y_i)/\sigma_{s y_k}$.
	Let~$\omega_i = y_\qr$ if~$k < i$, and~$\omega_i = y_0$ if~$k > i$.
	As stated above, for~$y_i$ to be on a shortest path from~$y_k$ to~$s \in \Vtwo(\Pmax_1)$, the path must visit~$\omega_i$.
	If~$s$ is in~$\Xleft_{y_k}$, then the shortest paths enter~$\Pmax_1$ via~$x_0$, and~$\sigma_{s y_k}(y_i)/\sigma_{s y_k} = \sigma_{x_0 y_k}(y_i)/\sigma_{x_0 y_k}$. 
	Note that there may be shortest~$s y_k$-paths that pass via~$y_0$ and~$s y_k$-paths that pass via~$y_\qr$.
	Thus we have
	\begin{align}
		\label{eq:bc-mp-inside-sigma1}
		\frac{\sigma_{x_0 y_k}(y_i)}{\sigma_{x_0 y_k}} = 
		\frac{\eta(y_k, \omega_i, x_0) \sigma_{x_0 \omega_i}}
		{\eta(y_k, y_0, x_0) \sigma_{x_0 y_0} + \eta(y_k, y_\qr, x_0) \sigma_{x_0 y_\qr}}.
	\end{align}
	With~$\sigma_{x_0 y_k}(y_i)$ we count the number of shortest~$x_0 y_k$-paths
	visiting~$y_i$.
	Note that any such path must visit~$\omega_i$.
	If there is such a shortest path visiting~$\omega_i$, then all shortest~$x_0 y_k$-paths
	visit~$y_i$, and since there is only one shortest~$\omega_i y_k$-path,
	the number of shortest~$x_0 y_k$-paths visiting~$\omega_i$ is equal to the number of
	shortest~$x_0 \omega_i$-paths, which is~$\sigma_{x_0 \omega_i}$.

	If~$s \in \Xright_{y_k}$, then
	\begin{align}
		\label{eq:bc-mp-inside-sigma2}
		\frac{\sigma_{s y_k}(y_i)}{\sigma_{s y_k}} =
		\frac{\eta(y_k, \omega_i, x_\qq) \sigma_{x_\qq \omega_i}}
		{\eta(y_k, y_0, x_\qq) \sigma_{x_\qq y_0} + \eta(y_k, y_\qr, x_\qq) \sigma_{x_\qq y_\qr}}.
	\end{align}
	Shortest paths from~$y_k$ to~$\xmid_{y_k}$ may visit any~$\varphi \in \{y_0, y_\qr\}$
	and~$\psi \in \{x_0 ,x_\qq\}$, and thus
	\begin{align}
		\label{eq:bc-mp-inside-sigma3}
		\frac{\sigma_{y_k \xmid_{y_k}}(y_i)}{\sigma_{y_k \xmid_{y_k}}}
		= \frac{\sum_{\psi \in \{x_0, x_\qq\}} \eta(y_k, \omega_i, \psi) \sigma_{\psi \omega_i}}
		{\sum_{\varphi \in \{y_0, y_\qr\} }\sum_{\psi \in \{x_0, x_\qq\}}
			\eta(y_k, y_\qr, \psi) \sigma_{\psi y_\qr}}.
	\end{align}
	Observe that
	\begin{compactenum}[\;\;\; (1)]
	\item the values of \Cref{eq:bc-mp-inside-sigma1,eq:bc-mp-inside-sigma2,eq:bc-mp-inside-sigma3} can be computed in constant time, since the values~$\sigma_{s t}$ are known for~$s, t \in \Vthree(G)$, and
	\item the values~$\sigma_{s y_k}(y_i)$ and~$\sigma_{s y_k}$ are independent of~$s$ for~$s \in \Xleft_{y_k}$ and for~$s \in \Xright_{y_k}$ respectively.
	\end{compactenum}
	Recalling that~$\Wleft[x_j] = \sum_{i=1}^j \Pend[x_i]$ and~$\Wright[x_j] = \sum_{i=j}^{\qq-1} \Pend[x_i]$ for~$1 \le j < \qr$ we define
	\begin{align*}
		\kappa(y_k, \omega_i)
		&= \Pend[y_k] \cdot \Big(\Pend[\xmid_{y_k}] \cdot
		\frac{\sum_{\psi \in \{x_0, x_\qq\}} \eta(y_k, \omega_i, \psi) \sigma_{\psi \omega_i}}
		{\sum_{\varphi \in \{y_0, y_\qr\} }\sum_{\psi \in \{x_0, x_\qq\}} \eta(y_k, \varphi, \psi) \sigma_{\psi \varphi}}\\
		&+ \sum_{s \in \Xleft_{y_k}} \Pend[s] \cdot \frac{\eta(y_k, \omega_i, x_0) \sigma_{x_0 \omega_i}} {\eta(y_k, y_0, x_0)\sigma_{x_0 y_0}+\eta(y_k, y_\qr, x_0) \sigma_{x_0 y_\qr}}\\
		&+ \sum_{s \in \Xright_{y_k}} \Pend[s] \cdot \frac{\eta(y_k, \omega_i, x_\qq) \sigma_{x_\qq \omega_i}} {\eta(y_k, y_0, x_\qq) \sigma_{x_\qq y_0}+\eta(y_k, y_\qr, x_\qq) \sigma_{x_\qq y_\qr}} \Big)
		\stepcounter{equation}\tag{\theequation}\label{eq:kappa}\\
		&=\Pend[y_k] \cdot \Big(\Pend[\xmid_{y_k}] \cdot\frac{\sum_{\psi \in \{x_0, x_\qq\}} \eta(y_k, \omega_i, \psi) \sigma_{\psi \omega_i}}{\sum_{\varphi \in \{y_0, y_\qr\} }\sum_{\psi \in \{x_0, x_\qq\}} \eta(y_k, \varphi, \psi) \sigma_{\psi \varphi}}\\
		&+ \Wleft[\xleft_{y_k}] \cdot \frac{\eta(y_k, \omega_i, x_0) \sigma_{x_0 \omega_i}} {\eta(y_k, y_0, x_0)\sigma_{x_0 y_0}+\eta(y_k, y_\qr, x_0) \sigma_{x_0 y_\qr}}\\
		&+ \Wright[\xright_{y_k}] \cdot \frac{\eta(y_k, \omega_i, x_\qq) \sigma_{x_\qq \omega_i}} {\eta(y_k, y_0, x_\qq) \sigma_{x_\qq y_0}+\eta(y_k, y_\qr, x_\qq) \sigma_{x_\qq y_\qr}} \Big).
	\end{align*}
	Note that since the values of~$\Pend[\cdot]$,~$\Wleft[\cdot]$ and of~$\Wright[\cdot]$ are known,~$\kappa(y_k, \omega_i)$ can be computed in constant time.

	If~$k < i$, then
	\begin{align*}
		\lambda(y_k, y_i)
		=&~\Pend[y_k] \cdot \Big(\Pend[\xmid_{y_k}] \cdot \frac{\sigma_{y_k \xmid_{y_k}}(y_i)}{\sigma_{y_k \xmid_{y_k}}}\\
		+&~\sum_{s \in \Xleft_{y_k}} \Pend[s] \cdot \frac{\sigma_{s y_k}(y_i)}{\sigma_{s y_k}} + \sum_{s \in \Xright_{y_k}} \Pend[s] \cdot \frac{\sigma_{s y_k}(y_i)}{\sigma_{s y_k}}\Big).
	\end{align*}
	\Cref{eq:bc-mp-inside-sigma1,eq:bc-mp-inside-sigma2,eq:bc-mp-inside-sigma3} then give us
	\begin{align*}
		\lambda(y_k, y_i)
		=&~\Pend[y_k] \cdot \Big( \Pend[\xmid_{y_k}] \cdot \frac{\sum_{\psi \in \{x_0, x_\qq\}} \eta(y_k, y_\qr, \psi) \sigma_{\psi y_\qr}}{\sum_{\varphi \in \{y_0, y_\qr\} }\sum_{\psi \in \{x_0, x_\qq\}}\eta(y_k, \varphi, \psi) \sigma_{\psi \varphi}}\\
		+&~\sum_{s \in \Xleft_{y_k}} \Pend[s] \cdot \frac{\eta(y_k, y_\qr, x_0) \sigma_{x_0 y_\qr}}{\eta(y_k, y_0, x_0)\sigma_{x_\qq y_0}+\eta(y_k, y_\qr, x_0)\sigma_{x_\qq y_\qr}}\\
		+&~\sum_{s \in \Xright_{y_k}} \Pend[s] \cdot \frac{\eta(y_k, y_\qr, x_\qq) \sigma_{x_\qq y_\qr}}{\eta(y_k, y_0, x_\qq) \sigma_{x_\qq y_0}+\eta(y_k, y_\qr, x_\qq) \sigma_{x_\qq y_\qr}}\Big)
		= \kappa(y_k, y_\qr).
	\end{align*}
	If~$k > i$, then analogously~$\lambda(y_k, y_i) = \kappa(y_k, y_0)$.
	Lastly, if~$k=i$, then~$\sigma_{s y_k}(y_i) = 0$; thus~$\gamma(s, y_k, y_i) = \lambda(y_k, y_i) = 0$.
	Hence, we can rewrite Sum~\eqref{eq:bc-mp-inside-base} as
	\begin{multline*}
		\sum_{k=1}^{\qr-1} \sum_{s\in\Vtwo(\Pmax_1)} \gamma(s, y_k, y_i)
		= \sum_{\substack{k=1\\k \ne i}}^{\qr-1} \lambda(y_k, y_i)
		= \Big(\sum_{k=1}^{i-1} \lambda(y_k, y_i) + \sum_{k=i+1}^{\qr-1} \lambda(y_k, y_i) \Big)\\
			= \Big(\sum_{k=1}^{i-1} \kappa(y_k, y_\qr) + \sum_{k=i+1}^{\qr-1} \kappa(y_k, y_0) \Big) \eqqcolon \rho_i.
	\end{multline*}

	Towards showing that Sum~\eqref{eq:A:bc-two-mp-inside} can be computed in~$O(r)$ time, note that~$\rho_0 = \sum_{k=1}^{\qr-1} \kappa(y_k, y_0)$ can be computed in~$O(|V(\Pmax_2)|)$ time.
	Observe that~$\rho_{i+1} = \rho_i - \kappa(y_{i+1}, y_\qr) + \kappa(y_i, y_0)$.
	Thus, every~$\rho_i$,~$1 \le i \le \qr$, can be computed in constant time.
	Hence, computing all~$\rho_i$,~$0 \le i \le \qr$, and thus computing sum~\eqref{eq:A:bc-two-mp-inside} for all~$v \in V(\Pmax_2)$ takes~$O(|V(\Pmax_2)|)$ time.
\end{proof}

\section{Proofs of Lemmata 9 and 10}
For the proofs of \cref{lem:v-in-one-maxpath,lem:v-not-within-maxpath} we first make two auxiliary observations and introduce some additional notation.

\begin{observation}
	\label{obs:dist-pmax}
	Let~$\Pmax = x_0 \dots x_\qq$ be a \maxpath{} and let~$0 \le i < j \le \qq$.
	Then
	\begin{enumerate}[(i)]
		\item
			$d_G(x_i, x_j) = \min \{ d_{\Pmax}(x_i, x_j), i + d_G(x_0, x_\qq) + q - j) \}$, and
		\item 
			if~$d_{\Pmax}(x_i, x_j) = i + d_G(x_0, x_\qq) + q - j$, then~$j = i + \frac{d_G(x_0,x_\qq) + \qq}{2}$.
	\end{enumerate}
\end{observation}
\begin{proof}
	The correctness of~(i) is clear.
	For~$(ii)$, note that the claimed equation is equivalent to~$j-i = d_{\Pmax}(x_i, x_j) = i + d_G(x_0, x_\qq) + q - j$.
\end{proof}

\begin{observation}
	\label{obs:sigma-pmax}
	Let~$\Pmax = x_0 \dots x_\qq$ be a \maxpath{}, let~$0 \le i < j \le \qq$, and let~$v \in V(G)$.
	Then
	\begin{align}
		\label{eq:onepathsigma}
		\frac{\sigma_{x_ix_j}(v)}{\sigma_{x_ix_j}} = \begin{cases}
			0, & \text{if } d_{\text{out}} < d_{\text{in}} \land v \in [x_i, x_j] \text{ or } d_{\text{in}} < d_{\text{out}} \land v \notin [x_i, x_j];\\
			1, & \text{if } d_{\text{in}} < d_{\text{out}} \land v \in [x_i, x_j];\\
			1, & \text{if } d_{\text{out}} < d_{\text{in}} \land  v \notin [x_i, x_j] \land v \in V(\Pmax);\\
			\frac{\sigma_{x_0x_\qq}(v)}{\sigma_{x_0x_\qq}}, & \text{if } d_{\text{out}} < d_{\text{in}} \land  v \notin V(\Pmax);\\
			\frac{1}{\sigma_{x_0x_\qq} + 1}, & \text{if } d_{\text{in}} = d_{\text{out}} \land  v \in [x_i, x_j];\\
			\frac{\sigma_{x_0x_\qq}}{\sigma_{x_0x_\qq} + 1}, & \text{if } d_{\text{in}} = d_{\text{out}} \land v \notin [x_i, x_j]\land  v \in V(\Pmax);\\
			\frac{\sigma_{x_0x_\qq}(v)}{\sigma_{x_0x_\qq} + 1}, & \text{if } d_{\text{in}} = d_{\text{out}} \land  v \notin V(\Pmax),
		\end{cases}
	\end{align}
	where~$d_{\text{in}} = d_{\Pmax}(x_i, x_j)$ and~$d_{\text{out}} = i + d_G(x_0, x_\qq) + q - j$.
\end{observation}
\begin{proof}
	Most cases are self-explanatory.
	The denominator~$\sigma_{x_0x_\qq} + 1$ is correct since there are~$\sigma_{x_0x_\qq}$ shortest paths from~$x_0$ to~$x_\qq$ (and therefore~$\sigma_{x_0x_\qq}$ shortest paths from~$x_i$ to~$x_j$ that leave~$\Pmax$) and one shortest path from~$x_i$ to~$x_j$ within~$\Pmax$.
	Note that if there are shortest paths that are not contained in~$\Pmax$, then~$d_G(x_0,x_\qq) <~\qq$ and therefore~$\Pmax$ is not a shortest~$x_0x_\qq$-path.
\end{proof}

\begin{definition}
	\label{def:i-j-mid}
	\emph{
	Let~$\Pmax = x_0 \dots x_\qq$ be a \maxpath{} and let~$0 \le i \le \qq$.
Then we define}
	\[ 
		i^+_{\Mid} = i + (d_G(x_0,x_\qq) + \qq)/{2} \quad\text{and}\quad j^-_{\Mid} = j - (d_G(x_0,x_\qq) + \qq)/{2}.
	\]
\end{definition}

\subsection{Proof of Lemma 9}
\label{sec:proofs-one-maxpath}
{
\def\thelemma{\ref{lem:v-in-one-maxpath}}
\begin{lemma}[Restated] 
	Let~$\Pmax = x_0 \ldots x_\qq$ be a \maxpath{}.
	Then, in~$O(|V(\Pmax)|)$ time, one can compute the following for all~$x_k$ with~$0 \le k \le \qq$:
	\[
		\alpha_{x_k} \coloneqq 2\cdot\sum_{i\in[1,k-1]}\sum_{j \in [k+1,\qq-1]}\gamma(x_i, x_j, x_k).
	\]
\end{lemma}
\addtocounter{lemma}{-1}
}
\begin{proof}
	We construct a dynamic program, then we show that it is solvable in~$O(|V(\Pmax)|)$ time.

	Note that~$1 \le i < k$.
	Thus for~$k = 0$ we have
	\[\alpha_{x_0} = 2\sum_{i\in\emptyset}\sum_{j \in [1,\qq-1]}\gamma(x_i, x_j, x_0) = 0.\] 
	This will be the base case of the dynamic program.

	For every vertex~$x_k$ with~$1\leq k < \qq$ it holds that
	\begin{multline*}
		\alpha_{x_k} = 2\cdot\sum_{\substack{i\in[1,k-1]\\j\in [k+1,\qq-1]}}\gamma(x_i, x_j, x_k)\\
		= 2\cdot\sum_{\substack{i\in[1,k-2]\\j\in [k+1,\qq-1]}}\gamma(x_i, x_j, x_k) + 2\cdot\sum_{j \in [k+1,\qq-1]}\gamma(x_{k-1}, x_{j}, x_k).
	\end{multline*}
	Similarly, for~$x_{k-1}$ with~$1 < k \leq \qq$ it holds that
	\begin{multline*}
		\alpha_{x_{k-1}} = 2\cdot\sum_{\substack{i\in[1,k-2]\\j\in [k,\qq-1]}}\gamma(x_i, x_j, x_{k-1})\\
		= 2\cdot\sum_{\substack{i\in[1,k-2]\\j\in [k+1,\qq-1]}}\gamma(x_i, x_j, x_{k-1}) + 2\cdot\sum_{i \in [1,k-2]}\gamma(x_{i}, x_{k}, x_{k-1}).
	\end{multline*}
	Next, observe that any path from~$x_i$ to~$x_j$ with~$i \leq k-2$ and~$j \geq k+1$ visiting~$x_k$ also visits~$x_{k-1}$ and vice versa.
	Substituting this into the equations above yields
	\[\alpha_{x_k} = \alpha_{x_{k-1}} + 2\cdot\sum_{j \in [k+1,\qq-1]}\gamma(x_{k-1}, x_{j}, x_k) -  2\cdot\sum_{i \in [1,k-2]}\gamma(x_{i}, x_{k}, x_{k-1}).\]

	Now we prove that~$\sum_{j \in [k+1,\qq-1]}\gamma(x_{k-1}, x_{j}, x_k)$ and~$\sum_{i \in [1,k-2]}\gamma(x_{i}, x_{k}, x_{k-1})$ can be computed in constant time once~$\Wleft$ and~$\Wright$ are precomputed (see \Crefrange{line:wleftright}{line:compute-wright} in \Cref{alg:overview}).
	These tables can be computed in~$O(|V(\Pmax)|)$ time as well.
	For the sake of convenience we say that~$\gamma(x_i,x_j,x_k) = 0$ if~$i$ or~$j$ are not integral or are not in~$[1,\qq-1]$ and define~$\Wtable[x_i,x_j] = \sum_{\ell = i}^{j} \Pend[x_\ell] = \Wleft[x_j]-\Wleft[x_{i-1}]$.
	Then we can use~\Cref{obs:dist-pmax,obs:sigma-pmax} to show that
	\begin{align*}
		& \sum_{j \in [k+1,\qq-1]}\gamma(x_{k-1}, x_{j}, x_k) = \sum_{j \in [k+1,\qq-1]} \Pend[x_{k-1}] \cdot \Pend[x_{j}] \cdot \frac{\sigma_{x_{k-1}x_{j}}(x_k)}{\sigma_{x_{k-1}x_{j}}}\\
		=&\ \gamma(x_{k-1}, x_{(k-1)^+_{\Mid}}, x_k) + \sum_{j \in [k+1,\min\{\lceil (k-1)^+_{\Mid}\rceil - 1, \qq-1 \}]} \Pend[x_{k-1}] \cdot \Pend[x_{j}]\\
		=& \begin{cases}
			\Pend[x_{k-1}] \cdot \Wtable[x_{k+1},x_{\qq-1}], \hfill \text{if } (k-1)^+_{\Mid} \geq \qq;\\
			\Pend[x_{k-1}] \cdot \Wtable[x_{k+1},x_{\lceil (k-1)^+_{\Mid}\rceil - 1}],  \ \ \ \ \ \text{if } (k-1)^+_{\Mid} < \qq\land (k-1)^+_{\Mid}\notin \mathds{Z};\\
			\Pend[x_{k-1}] \cdot (\Pend[x_{(k-1)^+_{\Mid}}] \cdot \frac{1}{\sigma_{x_0 x_\qq}+1} + \Wtable[x_{k+1},x_{(k-1)^+_{\Mid} - 1}]), \hfill \text{otherwise.}
		\end{cases}
	\end{align*}
	Herein we use the notation introduced in \cref{def:i-j-mid}.
	By~$(k-1)^+_{\Mid}\notin \mathds{Z}$ we mean to say that~$(k-1)^+_{\Mid}$ is not integral.
	Analogously,
	\begin{align*}
		& \sum_{i \in [1,k-2]} \gamma(x_{i}, x_{k}, x_{k-1}) = \sum_{i \in [1,k-2]} \Pend[x_{i}] \cdot \Pend[x_{k}] \cdot \frac{\sigma_{x_{i}x_{k}}(x_{k-1})}{\sigma_{x_{i}x_{k}}}\\
		= &\ \gamma(x_{k-1}, x_{k^-_{\Mid}}, x_{k-1}) + \sum_{i \in [\max\{1,\lfloor (k-1)^-_{\Mid}\rfloor + 1\}, k-2 ]} \Pend[x_{i}] \cdot \Pend[x_{k}]\\
		=& \begin{cases}
			\Pend[x_{k}] \cdot \Wtable[x_{1},x_{k-2}], \hfill \text{if } k^-_{\Mid} < 1;\\
			\Pend[x_{k}] \cdot \Wtable[x_{\lfloor k^-_{\Mid}\rfloor + 1}, x_{k-2}], \  \ \ \ \ \ \ \  \  \ \ \ \ \  \ \ \ \ \ \  \  \ \ \ \ \ \ \ \ \ \text{if } k^-_{\Mid} \geq 1\land k^-_{\Mid} \notin \mathds{Z};\\
			\Pend[x_{k}] \cdot (\Pend[x_{k^-_{\Mid}}] \cdot \frac{1}{\sigma_{x_0 x_\qq}+1} + \Wtable[x_{k^-_{\Mid} + 1},x_{k-2}]), \hfill \text{otherwise.}
		\end{cases}
	\end{align*}
	This completes the proof since~$(k-1)^+_{\Mid}$, $k^-_{\Mid}$, every entry in~$\Wtable[\cdot]$, and all other variables in the equation above can be computed in constant time once~$\Wleft[\cdot]$ is computed.
	Thus, computing~$\alpha_{x_i}$ for each vertex~$x_i$ in~\Pmax{} takes constant time.
	Hence, the computations for the whole \maxpath{}~$\Pmax$ take~$O(|V(\Pmax)|)$ time.
\end{proof}

\subsection{Proof of Lemma 10}
\label{sec:proofs-one-maxpath-not-within}
{
\def\thelemma{\ref{lem:v-not-within-maxpath}}
\begin{lemma}[Restated] 
	Let~$\Pmax = x_0 x_1 \dots x_\qq$ be a \maxpath{}.
	Then, assuming that~$\Inc[s, t] \cdot \sigma_{s t}(v)$ can be computed in constant time for some~$s, t \in \Vthree(G)$, one can compute in~$O(|V(\Pmax)|)$ time the following for all~$v \in V(G) \setminus [x_i, x_j]$:
 	\[\beta_v := \sum_{i\in[1,\qq-1]}\sum_{j \in [i+1,\qq-1]}\gamma(x_i, x_j, v).\]
\end{lemma}
\addtocounter{lemma}{-1}
}
\begin{proof}
	We first show how to compute~$\beta_v$ for all~$v \notin V(\Pmax)$ and then how to compute~$\beta_v$ for all~$v \in V(\Pmax) \setminus [x_i, x_j]$ in the given time.

	As stated above, the distance from~$x_i$ to~$x_{i^+_{\Mid}}$ (if existing) is the boundary such that all shortest paths to vertices~$x_j$ with~$j > i^+_{\Mid}$ leave~$\Pmax$ and the unique shortest path to any~$x_j$ with~$i < j < i^+_{\Mid}$ is~$x_ix_{i+1}\ldots x_j$.
	Thus we can use \cref{obs:dist-pmax,obs:sigma-pmax} to show that for each~$v\notin \Pmax$ and each fixed~$i \in [1,\qq-1]$ it holds that 
	\begin{align*}
		&\sum_{j \in [i+1,\qq-1]}\gamma(x_i, x_j, v) = \sum_{j \in [i+1,\qq-1]} \Pend[x_i] \cdot \Pend[x_j] \cdot \frac{\sigma_{x_ix_j}(v)}{\sigma_{x_ix_j}}\\
		=&\begin{cases}
			0, \hfill\text{if } i^+_{\Mid} > \qq-1;\\
			\sum_{j \in [x_{\lceil i^+_{\Mid}\rceil},\qq-1]} \Pend[x_i] \cdot \Pend[x_j] \cdot \frac{\sigma_{x_0x_\qq}(v)}{\sigma_{x_0x_\qq}},\qquad \text{if } i^+_{\Mid} \leq \qq-1 \land i^+_{\Mid} \notin \mathds{Z};\\
			\Pend[x_i] \cdot \Big ( \Pend[x_{i^+_{\Mid}}] \cdot \frac{\sigma_{x_0x_\qq}(v)}{\sigma_{x_0x_\qq}+1} + \sum_{j \in [x_{i^+_{\Mid} +1},\qq-1]} \cdot \Pend[x_j] \cdot \frac{\sigma_{x_0x_\qq}(v)}{\sigma_{x_0x_\qq}} \Big ),\\
			\hfill\text{otherwise;}\\
		\end{cases}\\
		=&\begin{cases}
			0, \hfill\text{if } i^+_{\Mid} > \qq-1;\\
			\Pend[x_i] \cdot \Wright[x_{\lceil i^+_{\Mid}\rceil}] \cdot \frac{\sigma_{x_0x_\qq}(v)}{\sigma_{x_0x_\qq}}, \hfill \text{if } i^+_{\Mid} \leq \qq-1 \land i^+_{\Mid}\notin \mathds{Z};\\
			\Pend[x_i] \cdot \Big ( \Pend[x_{i^+_{\Mid}}] \cdot \frac{\sigma_{x_0x_\qq}(v)}{\sigma_{x_0x_\qq}+1} + \Wright[x_{i^+_{\Mid} + 1}] \cdot \frac{\sigma_{x_0x_\qq}(v)}{\sigma_{x_0x_\qq}} \Big ),\qquad\text{otherwise.}\\
		\end{cases}
	\end{align*}
	Herein we use the notation introduced in \cref{def:i-j-mid}.
	By~$i^+_{\Mid} \notin \mathds{Z}$ we mean to say that~$i^+_{\Mid}$ is not integral.
	All variables except for~$\sigma_{x_0x_\qq}(v)$ can be computed in constant time once \Wright{} and~$\sigma_{x_0x_\qq}$ are computed.
	Thus we can compute overall in~$O(|V(\Pmax)|)$ time the value
	\begin{multline}
		\label{eq:onepathX}
		X = \frac{2\cdot\sum_{i\in[1,\qq-1]}\sum_{j \in [i+1,\qq-1]}\gamma(x_i, x_j, v)}{\sigma_{x_0x_\qq}(v)}\\
		= 2\cdot\sum_{i\in[1,\qq-1]}\sum_{j \in [i+1,\qq-1]}\Pend[x_i]\Pend[x_j]\sigma_{x_i,x_j}.
	\end{multline}
	Due to the postprocessing (see \Cref{line:modified-bfs-loop,line:modified-bfs} in \Cref{alg:overview}) it is sufficient to add~$X$ to~$\Inc[x_0,x_\qq]$.
	This ensures that $X \cdot \sigma_{x_0x_\qq}(v)$ is added to the \bc{} of each vertex~$v \notin V(\Pmax)$.
	Note that if~$X > 0$, then~$d_G(x_0,x_\qq) < \qq$ and thus the \bc{} of any vertex~$v \in V(\Pmax)$ is not affected by~$\Inc[x_0,x_\qq]$.

	Next, we will compute~$\beta_v$ for all vertices~$v \in V(\Pmax)$ (recall that~$v \notin [x_i,x_j]$).
	We start with the simple observation that all paths that leave~$\Pmax$ at some point have to contain~$x_0$.
	Thus~$\beta_{x_0}$ is equal to~$X$ by \Cref{eq:onepathX}.
	We will use this as the base case for a dynamic program that iterates through~$\Pmax$ and computes~$\beta_{x_k}$ for each vertex~$x_k$,~$k\in[0,\qq]$, in constant time.

	Similarly to the proof of \Cref{lem:v-in-one-maxpath} we observe that 
	\begin{align*}
		\beta_{x_k} =&\ 2\Big (\sum_{i\in[k+1,\qq-1]} \sum_{j\in [i+1,\qq-1]}\gamma(x_i, x_j, x_k) &+& \sum_{i\in[1,k-1]} \sum_{j\in [i+1,k-1]}\gamma(x_i, x_j, x_k) \Big )&\\
		=&\ 2\Big (\sum_{i\in[k+2,\qq-1]} \sum_{j\in [i+1,\qq-1]}\gamma(x_i, x_j, x_k) &+& \sum_{i\in[1,k-1]} \sum_{j\in [i+1,k-1]}\gamma(x_i, x_j, x_k)\\
		& &+& \sum_{j\in [k+2,\qq-1]}\gamma(x_{k+1}, x_j, x_k) \Big )
	\end{align*} 
	and
	\begin{align*}
		\beta_{x_{k+1}} &= 2\Big (\sum_{i\in[k+2,\qq-1]} \sum_{j\in [i+1,\qq-1]}\gamma(x_i, x_j, x_{k+1}) + \sum_{i\in[1,k]} \sum_{j\in [i+1,k]}\gamma(x_i, x_j, x_{k+1}) \Big )&\\
		&= 2\Big (\sum_{i\in[k+2,\qq-1]} \sum_{j\in [i+1,\qq-1]}\gamma(x_i, x_j, x_{k+1})\\
		&\quad + \sum_{i\in[1,k-1]} \sum_{j\in [i+1,k-1]}\gamma(x_i, x_j, x_{k+1}) + \sum_{i\in [1,k-1]}\gamma(x_{i}, x_k, x_{k+1}) \Big ).
	\end{align*}
	Furthermore, observe that every~$st$-path with~$s,t \neq x_k,x_{k+1}$ that contains~$x_k$ also contains~$x_{k+1}$, and vice versa.
	Thus we can conclude that
	\[
		\beta_{x_{k+1}} = \beta_{x_k} + 2\Big ( \underbrace{\sum_{i\in [1,k-1]}\gamma(x_{i}, x_k, x_{k+1})}_{(*)} - \underbrace{\sum_{j\in [k+2,\qq-1]}\gamma(x_{k+1}, x_j, x_k)}_{(**)} \Big ).
	\]

	It remains to show that the sums~$(*)$ and~$(**)$ can be computed in constant time once~$\Wleft$ and~$\Wright$ are computed.
	Using \Cref{obs:dist-pmax,obs:sigma-pmax} we get that
	\[
		\sum_{i\in [1,k-1]}\gamma(x_{i}, x_k, x_{k+1})=
		\begin{cases}
			0, \hfill \text{if } k^-_{\Mid} < 1;\\
			\Pend[x_k] \cdot \Wleft[x_{\lfloor k^-_{\Mid}\rfloor}], \qquad \text{if } k^-_{\Mid} \geq 1 \land k^-_{\Mid}\notin \mathds{Z};\\
			\Pend[x_k] \cdot \Big ( \Wleft[x_{k^-_{\Mid} - 1}] + \Pend[k^-_{\Mid}]\cdot \frac{\sigma_{x_0x_\qq}}{\sigma_{x_0x_\qq} + 1}\Big ),\\
			\hfill \text{otherwise};\\
		\end{cases}
	\]
	and
	\begin{align*}
		\sum_{j\in [k+2,\qq-1]}&\gamma(x_{k+1}, x_j, x_k)\\
		=&
		\begin{cases}
			0, \qquad\qquad\qquad\qquad\qquad\qquad\qquad\qquad\qquad\qquad\qquad\;\;\; \text{if } k^+_{\Mid} < 1;\\
			\Pend[x_{k+1}] \cdot \Wright[x_{\lceil (k+1)^+_{\Mid}\rceil}],\\
			\hfill \text{if } (k+1)^+_{\Mid} \leq \qq-1 \land (k+1)^+_{\Mid}\notin \mathds{Z};\\
			\Pend[x_{k+1}] \cdot \Big ( \Wright[x_{(k+1)^+_{\Mid} + 1}] + \Pend[(k+1)^+_{\Mid}]\cdot \frac{\sigma_{x_0x_\qq}}{\sigma_{x_0x_\qq} + 1}\Big ),\\
			\hfill \text{otherwise.}\\
		\end{cases}
	\end{align*}
	Since all variables in these two equalities can be evaluated in constant time, this concludes the proof.
\end{proof}

\end{document}